\documentclass[11pt]{amsart}

\usepackage[utf8]{inputenc}
\usepackage{amsmath}
\usepackage{amsthm}

\usepackage{lmodern}
\usepackage{anyfontsize}

\usepackage{bm}
\usepackage{bbold}
\usepackage{thm-restate} 

\usepackage{subfig}
\usepackage{graphicx}
\usepackage{float}
\usepackage{esdiff}
\usepackage{comment}
\usepackage{mathtools} 
\usepackage{xcolor} 
\usepackage{bbm} 
\usepackage{dsfont} 
\usepackage{diagbox}
\usepackage{pdfsync}

\usepackage{biblatex}
\usepackage{amsfonts}
\usepackage{amssymb}

\usepackage{url}
\usepackage{doi}
\usepackage{color}

\usepackage{tikz}
\usepackage{todonotes}

\usepackage{bm}
\usepackage{bbold}

\usepackage{subfig}
\usepackage{graphicx}
\usepackage{float}
\usepackage{esdiff}
\usepackage{comment}
\usepackage{mathtools} 
\usepackage{xcolor} 
\usepackage{bbm} 
\usepackage{dsfont} 
\usepackage{diagbox}
\usepackage{pdfsync}
\newtheorem{theorem}{Theorem}
\newtheorem{prop}{Proposition}
\newtheorem{cor}{Corollary}
\newtheorem{lem}{Lemma}

\theoremstyle{definition}
\newtheorem{dfn}{Definition}

\newtheorem{claim}{Claim}
\newtheorem{rem}{Remark}

\newcommand{\R}{\mathbb{R}}

\newcommand{\PP}{\mathbb{P}}
\newcommand{\E}{\mathbb{E}}

\usepackage{amsfonts}
\usepackage{amssymb}

\title{Markovian protocols and an upper bound on the extension complexity of the matching polytope}

\title{Markovian protocols and an upper bound on the extension complexity of the matching polytope}
\author[M. Szusterman]{M. Szusterman}

\address{Centre de Mathématiques Laurent Schwartz (CMLS) \\
École Polytechnique, 91128 Palaiseau Cedex, France}
\email{maud.szusterman@polytechnique.edu}

\begin{document}

\begin{abstract}
This paper investigates the extension complexity of polytopes by exploiting the correspondence between non-negative factorizations of slack matrices and randomized communication protocols. We introduce a geometric characterization of extension complexity based on the width of Markovian protocols, as a variant of the framework introduced by Faenza et al. This enables us to derive a new upper bound of $\tilde{O}(n^3\cdot 1.5^n)$ for the extension complexity of the matching polytope $P_{\text{match}}(n)$, improving upon the standard $2^n$-bound given by Edmonds' description. Additionally, we recover Goemans' compact formulation for the permutahedron using a one-round protocol based on sorting networks.
\end{abstract}

\maketitle


\section{Introduction}

An extended formulation of a polytope $P$ represents it as the linear projection of a higher-dimensional polytope $Q$. These objects are central to convex optimization and polyhedral combinatorics : a convex optimization problem over $P$ can be lifted to $Q$, which is advantageous if $Q$ admits a significantly more succinct description than $P$ (in the latter case, $Q$ is called a \emph{compact formulation} of $P$). Given a description $P=\{x\in \R^d : Ax\leq b , Cx=e\}$ of $P$, with $A\in \R^{m\times d}$, $C\in \R^{m'\times d}$, $b\in \R^m$, $e\in \R^{m'}$, we call $m$ the size of the description. The size of $P$, denoted $|P|$ in this work, is by definition the minimal size of such a description. Geometrically, $|P|$ is the number of facets of the polytope. The \emph{extension complexity} of $P$ is, by definition, the least size of an affine extension $Q$ of $P$ :
\begin{equation}
\label{dfnxcP}
\text{xc}(P):=\inf \{|Q| : Q \text{ is an affine extension of }P\}
\end{equation}
that is, the minimal number of inequalities needed to describe $P$, if (extra variables and) affine sections and projections are allowed.
Given an extended formulation $(Q,\pi)$ of $P$, and a convex function $f$ defined on $P$, a way to find $x^*\in\text{argmin}(f(x), x\in P)$ is to first find some $z^*\in \text{argmin}(\overline{f}(z), z\in Q)$ optimal for the (lifted)\footnote{(if for instance we try to minimize $f(x)=\langle c,x\rangle$ over $P$, then ``lifting'' the problem to $Q$ is just saying:
$\min 
\{\langle c,x\rangle , x\in P\}=\min 
\{\langle c,\pi z\rangle , z\in Q\}=\min 
\{\langle \pi^*c, z\rangle , z\in Q\} \hspace{2mm}).$ } problem, and then project $z^*$ onto $P$ ; since some optimization algorithms (such as the interior point method) have running time depending on $m$, the number of affine constraints, if $|Q|<<|P|$ then it is much faster to look for $z^*$ (and project) rather than for $x^*$ directly.



For instance, consider the permutahedron $\text{Perm}(n)$, the convex hull of the $n!$ vectors $x_{\sigma}=(\sigma(1),\sigma(2),...,\sigma(n))\in \R^n$, $\sigma\in \mathcal{S}_n$. While this polytope has dimension $n-1$, it has $2^n-2$ facets, which are given by Edmonds description :
\begin{equation}
\label{edmondsreppermu}
\text{Perm}(n)= \left\{ x\in \R^n : x([n])=\frac{n(n+1)}{2} ; x(J) \geq \frac{|J|(|J|+1)}{2} \quad  (\forall J \subsetneq [n], J\neq \emptyset) \right\} \quad,
\end{equation}
(one can check that none of these inequalities could be removed).

Now, consider the ($n^{th}$) Birkhoff polytope : it is defined as the set of doubly stochastic matrices, and, denoting $M_{\sigma}$ the permutation matrices :
\begin{align*}
\text{Birk}(n) &=\{M\in [0,1]^{n\times n} : \forall i\in [n], \sum_j M_{i,j}=\sum_j M_{j,i}=1\} =\text{conv}\{M_{\sigma} | \sigma\in\mathcal{S}_n\}.
\end{align*}
Equality of these two representations is the Birkhoff-von-Neumann theorem.

Let us stare for a minute at these two polytopes, and let us call $P=\text{Perm}(n)$ and $Q=\text{Birk}(n)$. Their dimensions differ quadratically : $\text{dim}(P)=n-1$ and $\text{dim}(Q)=(n-1)^2$, but their sizes differ exponentially : $|P|=2^n-2$ and $|Q|=n^2$. There is a clear bijection between their sets of vertices, and this bijection is in fact induced by the linear projection $\pi: \mathbb{R}^{n \times n} \to \mathbb{R}^n$ defined by $\pi(M) = M \cdot (1, 2, \dots, n)^\top$: one easily checks that $\pi \text{Birk}(n)=\text{Perm}(n)$.
Hence $\text{Perm}(n)$ is an example of polytope admitting an extension whose descriptive size is exponentially smaller than that of $P$.

Given a polytope $P\subset \R^d$, for which we have the two dual representations $P=\text{conv}(x_1, \ldots , x_N)=\{x\in H / Ax\leq b\}$ with $H$ some affine subspace of $\R^d$ and $(A,b)\in \R^{m\times d} \times \R^{m}$, one defines $S$, the\footnote{canonically, both $N$ and $m$ will be minimal, corresponding respectively to the number of vertices of $P$, and its number of facets, but extra points $x_j$, and extra constraints $(a_i,b_i)$ (corresponding to lower-dimensional faces $F$ of $P$) may be included.}  slack matrix of $P$, by $S_{i,j}=b_i-\langle a_i, x_j\rangle$, where $a_i\in \R^d$ is the $i^{th}$ row of $A$. In particular $S_{i,j}\geq 0$ for all $(i,j)\in [m]\times [N]$. Therefore one can define the non-negative rank of this matrix:
$$\text{rk}_+(S):=\inf\{r\geq 0 : \exists U\in \R_{\geq 0}^{m\times r}, \exists V\in \R_{\geq 0}^{r\times N}, S=UV\} \quad.$$
Farkas lemma ensures that $\text{rk}_+(S)$ doesn't depend on which slack matrix of $P$ was chosen, hence one can denote $\text{rk}_+(P):=\text{rk}_+(S)$. 


A seminal result in the field is the following characterization of extension complexity.
\begin{theorem}[Yannakakis]
\label{yannak}
Let $P\subset \R^d$ be a polytope. Then \emph{xc}$(P)=\text{rk}_+(P).$
\end{theorem}
This allows for example to get lower bounds on $\text{xc}(P)$ by finding lower bounds on $\text{cov}(S)$, the covering number of slack matrix $S$. This was used for instance in \cite{fiorinibirkhoff} to show that $\text{xc}(B_n)=|B_n|=n^2$, if $B_n=\text{Birk}(n)$ (by constructing an adequate fooling set of size $n^2$ for the slack matrix of $B_n$), or by Kaibel and Weltge  to reprove that $\text{CORR}(n)$ has extension complexity at least $1.5^n$ (using Razborov's lemma, \cite{kaibelw2015}).   In the other direction, upper bounding $\text{rk}_+(S)$, say by finding a compact (non-negative) factorization of the slack matrix, translates into an upper bound on $\text{xc}(P)$. By exploiting a correspondence between certain \emph{randomized protocols}, which, given any input $(i,j)\in [m]\times [N]$, produce a non-negative output $w$, whose average value matches the target entry:  $\E_{A,B}(w| A\leftarrow i, B\leftarrow j)=S_{i,j}$, and non-negative factorizations of a slack matrix $S$,  Faenza, Fiorini, Grappe and Tiwary have proven yet another characterization of extension complexity:
\begin{theorem}[FFGT -\cite{faenzaetal}]
\label{faenzacharacterization}
For any polytope $P$:
\[\lceil\log_2(\text{\emph{xc}}(P)\rceil=\min\{\text{cost}(\pi) : \pi \text{ (non-neg.) protocol computing }S, \text{ on average}\}.\]
\end{theorem}
A usual representation of a protocol is via binary trees (the same tree being used for all entries), then the cost (number of bits exchanged between the two parties, in the worst-case) is the height of the tree supporting the protocol. This explains the need for a $\log_2$ in the above statement\ref{faenzacharacterization}. In this work, we choose a different representation of (randomized) protocols, which is most relevant for \emph{markovian} protocols (see paragraphe \ref{para:rdmprotocols} for a  definition), and whose structure underlines the number of rounds of the protocol (which will be the number of layers of the BP). This choice enables a rewriting of Theorem \ref{faenzacharacterization} as follows:
\begin{theorem}
\label{maudversionfaenzathm}
Let $P\subset \R^d$ be a polytope, and $S$ a slack matrix of $P$.  Then 
\[\text{\emph{xc}}(P)=\min\{|\Gamma(\pi)| : \pi \text{\emph{ markovian }protocol computing } S \text{ on average}\}\]
where $\Gamma(\pi)$ is a certain geometric structure carrying the protocol.
\end{theorem} 
Informally, $\Gamma(\pi)$ is a certain subset of paths, defined using the BP supporting the protocol $\pi$, and according to the probability distributions used by the two parties for their private computations ran during the protocol. We shall call $|\Gamma(\pi)|$ the \emph{width}\footnote{this terminology is borrowed from branching program literature} of the protocol. See page 5.

Using this, and revisiting some combinatorial arguments developed in \cite{faenzaetal} which allowed them to derive $\text{xc}(P_{PM}(n))\leq n^3 2^{n/2}$ for the perfect matching polytope (if $n$ even), we can prove the following upper bound (see paragraph \ref{section:PMG} for the definition of $P_{\text{match}}(G)$, and for a dual characterization of this polytope).
\begin{theorem}
\label{mainresult}
Let $G=(V,E)$ be a graph on $|V|=n$ vertices. Denote $P_{\text{match}}(G)\subset \R^{E}$ the matching polytope of $G$. Then $\text{xc}(P_{\text{match}}(G))\leq \ln(n) \cdot n^3 \cdot 1.5^n$.
\end{theorem}
To the best of our knowledge, this is the first upper bound on $\text{xc}(P_{\text{match}}(K_n))$ which falls below $2^n$.
In fact, matching and perfect matching polytopes are intimately related: for any $G$, $P_{PM}(G)$ is a face of $P_{\text{match}}(G)$, as $P_{PM}(G)=H\cap P_{\text{match}}(G)$, with $H=\{x\in \R^E : x(E)=|V|/2\}$, but that allows to deduce upper bounds on $\text{xc}(P_{PM}(G))$ from these on $P_{\text{match}}(G)$, not the reverse. In the other direction, $P_{\text{match}}(G)$ can be seen as the projection of a face\footnote{(the face is the intersection of $P_{PM}(\overline{G})$ with $H=\{x\in \R^{\overline{E}} : u\neq v \Rightarrow x_{uv'}=0\}$, and the projection is simply $\pi : \R^{\overline{E}}=\R^{E} \times \R^{E'}\times \R^{V\times V'} \to \R^{E}$ which keeps the $E$-coordinates)} of $P_{PM}(\overline{G})$, where $\overline{G}=(\overline{V},\overline{E})$, with $\overline{V}:=V\cup V'$ ($V'$ a disjoint copy of $V$), and with $\overline{E}:=E \cup E' \cup (V\times V')$, for instance $\overline{G}=K_{2n}$ if $G=K_n$. This implies that $\text{xc}(P_{\text{match}}(G))\leq \text{xc}(P_{PM}(\overline{G}))$; in particular $\text{xc}(P_{\text{match}}(n))\leq \text{xc}(P_{PM}(2n)) \leq (2n)^3 2^n$ (the latter is Faenza et al.'s upper bound \cite{faenzaetal}), but this isn't too helpful for $\text{xc}(P_{\text{match}}(n))$, since $|P_{\text{match}}(n)| \leq 2^{n-1}+\frac{n(n-1)}{2}$, as given by Edmonds description of $P_{\text{match}}(n)$ ( \cite{Edmonds65matchings}).

Furthermore, again using the concept of randomized protocols developed in \cite{faenzaetal}, we present a one-round protocol based on sorting networks which allows to recover Goemans’ upper bound for the extension complexity of the permutahedron. In fact the non-negative factorization arising from this protocol is the one suggested by Fiorini et al. in \cite{FRT12}. 
%
We detail this one-round protocol in Section \ref{section:permu}, and prove its correctness. Links with Goemans' direct construction $Q_n\subset \R^{n+2q}$ are discussed (\ref{para:comparisongoemans}). We leave some open questions regarding minimal sorting networks (and their resulting factorizations).
%

\vspace{3mm}
\emph{Remerciements}. I thank Balthazar Bauer and Hervé Fournier for valuable discussions.

\section{Markovian protocols with branching programs}
\label{para:rdmprotocols}

A Markovian protocol in $k$ rounds of communication is defined as follows.

Given $k$ disjoint sets $V_1, \ldots , V_k$, and sets of directed edges $\overrightarrow{E}_j \subset V_j\times V_{j+1}$, for $1\leq j\leq k-1$, and $\overrightarrow{E}_0 =\{s\} \times V_1$ and $\overrightarrow{E}_f \subset V_k \times \{t\}$, one can use the branching program structure $\mathcal{A}=(\{s\},V_1, \ldots , V_k, \{t\}, \overrightarrow{E}_j, j\in [k])$, to produce random outputs. More precisely, given $\mathcal{X}$ and $\mathcal{Y}$ two sets, we assume that for each $x\in \mathcal{X}$, there is an initial probability distribution $p^0_{A,x}\in \mathcal{P}(V_1)$ (which tells Alice how to pick $u_1\in V_1$), and for $1\leq j\leq k-1$, transition probabilities $p^{(j)}_{D,z}(u_j, *)\in \mathcal{P}(V_{j+1})$, with $(D,z)=(A,x)$ if $j$ even (resp. $(D,z)=(B,y)$ if $j$ odd), which tells party $D$ how to pick $u_{j+1}\in V_{j+1}$, given his/her initial input $z$, and knowing the last recevied message from $B$ (resp. from $A$) is $u_j$. Finally, say $k$ is even, then party $A$, once receiving message $u_k$ (from $B$) claims a random output $\omega_{A}=\omega_{A,x}(u_k)$, chosen according to $p^{(k)}_{A,x,u_k}$, a probability distribution on $\mathcal{P}([0,\infty)$.\footnote{Similarly, if $k$ is odd, then $B$ shall claim $\omega_B:=\omega_{B,y}(u_k)$, which is random and drawn (by B) with respect to a certain distribution $p^{(k)}_{B,y,u_k}\in \mathcal{P}([0,\infty))$. )}

Thus $\mathcal{A}=(\{s\},V_1, \ldots , V_k, \{t\}, \overrightarrow{E}_j, j\in [k])$, together with the initial distribtion $p^0_{A,x}$, the transition probabilities $p^{(j)}_{D,z,u_j}\in \mathcal{P}(V_{j+1})$ (with $(D,z)\in \{(A,x),(B,y)\}$ depending on parity of $j$), and the random output $\omega_{D}$ claimed by A (or B) according to her received initial input, and to last received message $u_k$, defines a two-party communication protocol in $k$ rounds of communication.

Remark :  of course, one can also define protocols which start with B as the first speaker, i.e. with initial distributions $p^0_{B,y}\in \mathcal{P}(V_1)$, then transition probabilities $p^{(j)}_{D,z}(u_j,u_{j+1})$ with $(D,z)=(A,x)$ if $j$ odd, and$(D,z)=(B,y)$ if $j$ even, etc. Similarly, the definition which we give just below can also be stated for protocols with either A or B as a first speaker, and with either A or B as a final claimer. For brievity, we only state one of the four definitions.

\begin{dfn}
\label{dfnproto}
Given a two-party communication protocol supported by an BP $\mathcal{A}=(\{s\},V_1, \ldots , V_k, \{t\}, \overrightarrow{E}_j, j\in [k])$, say $A$ is the first to speak, and say she is also the (final) claimer. Then the average output of this protocol, when A is given x for input, and B is given y for input, is:
\[\E_{A,B}[\omega_A | A\leftarrow x, B\leftarrow y]=\sum_{u_1, \ldots , u_k} p^0_{A,x}(u_1)p^{(1)}_{B,y}(u_1,u_2) p^{(2)}_{A,x}(u_2,u_3) \cdots p^{(k-1)}_{B,y}(u_{k-1},u_k) \E_A[\omega_{A,x,u_k}] \]
where $\E_A[\omega_{A,x,u_k}]=\E_{\omega \sim p^{(k)}_{A,x,u_k}}[\omega_A]$. In case the choice of the output $\omega_A$ is deterministic given $u_k$, we may more simply write:
\[\E_{A,B}[\omega_A | A\leftarrow x, B\leftarrow y]=\sum_{u_1, \ldots , u_k} p^0_{A,x}(u_1)p^{(1)}_{B,y}(u_1,u_2) p^{(2)}_{A,x}(u_2,u_3) \cdots p^{(k-1)}_{B,y}(u_{k-1},u_k) \omega_{A,x}(u_k) \]
\end{dfn}
One can also define the average value of the ouput \emph{knowing} that the first shared messages have been $(u_1,\ldots , u_j)$ (for some $j<k$). Say $j$ is even, and $A$ speaks first, so the last message $u_j$ was sent by Bob, then (say $k$ is odd) this value is
\[\E_{A,B}[\omega_B | A\leftarrow x, u_j ; B\leftarrow y] = \sum_{u_{j+1}, \ldots , u_k} p^{(j)}_{A,x}(u_j,u_{j+1}) \cdots p^{(k-1)}_{A,x}(u_{k-1}, u_k) \E_B[\omega_B | B\leftarrow y, u_k]. \]

Note that the average output knowing $(u_1, \ldots , u_j)$ is the same as the average output if only knowing $u_j$ (this is why such protocols are called markovian).

Now protocols are usually designed to help the 2 parties find (or guess) the value $f(x,y)$, where $f: \mathcal{X} \times \mathcal{Y} \to \R$ is a given function, known by both parties, and where each party only knows part of the input. When shall a randomized protocol considered  helpful regarding a target matrix $f$ ? We follow \cite{faenzaetal}, and translate to our model their notion of correctness (for randomized protocols).


\begin{dfn}
\label{dfncorrect}
Let $\pi$ be a  protocol for A and B, aimed at guessing the value $f(x,y)$, when $A$ knows $x$ and $B$ knows $y$. Assume the protocol is a Markovian protocol in $k$ rounds, i.e. with its structure supported by a BP in $k$ layers, as described above. We shall say that $\pi$ is \emph{correct} if \footnote{in words, $\pi$ is correct when for any inputs $(x,y)$, the average value $\omega_A$ (or $\omega_B$) claimed by one or the other party at the end of the protocol, matches $f$ exactly} : 
\[\forall (x,y)\in \mathcal{X} \times \mathcal{Y}, \quad \E_{A,B}[\omega | A\leftarrow x, B\leftarrow y] =f(x,y). \]
\end{dfn}

We may have further requirements on the protocol, eg. asking that the outputs always remain within a given domain. For instance if  $f: \mathcal{X} \times \mathcal{Y} \to [0,\infty)$, we will say a protocol $\pi$ is correct, if given any inputs, the average value of the output matches $f$, and if $\PP_{A,B}(\omega \geq 0)=1$. 

To illustrate the above notations, let us quote an example from \cite{faenzaetal}. Given a (connected) graph $G=(V,E)$, denote $SpT(G)$ the set of its spanning trees $SpT(G)=\{T\subset E / |T|=|V|-1, T \text{ tree }\}$, and denote $P_{SpT}(G)$ the spanning tree polytope of $G$. If $|V|=n$, then it is given by (see \cite{martin91}):
\[ P_{SpT}(G)=\text{conv}\{ \chi_T : T\in SpT(G)\}=\{x\in \R_{\geq 0}^E : x(E)=n-1; x(E(U))\leq |U|-1, \hspace{2mm}\forall U \subsetneq V \} .\]


Call $\mathcal{X}=\{U\subset V : U\neq V, E(U) \neq \emptyset\}$ and call $\mathcal{Y}=SpT(G)$. For each pair $(x,y)=(U,T)\in \mathcal{X} \times \mathcal{Y}$, one can define slack $f(U,T)=|U|-1-\chi_T(E(U))=|U|-1-|T\cap E(U)| \geq 0$. In \cite{faenzaetal}, the following protocol is given, to \emph{guess} the value $f(x,y)$.
So Alice receives $U$, and Bob receives $T$, as respective inputs. It is a two-round protocol ($k=2$ in Definition \ref{dfnproto}).

\begin{itemize}
\item Alice picks $u_0\in U$ uniformly at random ($u_0\sim \text{Unif}(U)$). She sends $u_0$ to B.
\item Bob picks $e\sim \text{Unif}(T)$. He orients $e$ towards $u_0$ (on $T$), and sends $\vec{e}=(u,v) \in \overrightarrow{E}$ to A.
\item Alice claims $\omega_A=(n-1)$ if ($u\in U$ and $v\notin U$), and claims $0$ otherwise.
\end{itemize}

\begin{figure}{}
\centering
\includegraphics[width=0.6 \textwidth]{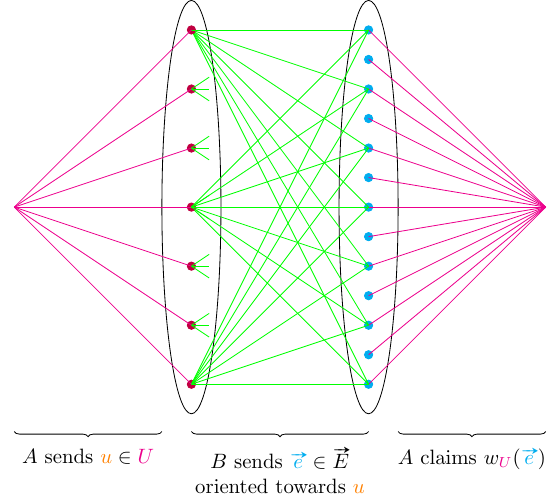}
\caption{Quoting a Markovian protocol for the spanning tree polytope}
\end{figure}

In this example, we have $V_1=V$ (the first layer of the BP is the set of vertices of $G$), and $V_2=\overrightarrow{E}$ : the second layer is the set of oriented edges of the graph. Moreover, the transition probabilities are given by :
$p^0_{A,U}=\text{Unif}(U)$, i.e. $p^0_{A,U}(v)=\frac{1}{|U|} \chi_U(v) \in \mathcal{P}(V)$.
 If $(x,y)\in \overrightarrow{E}$, we denote $xy\in E$ the corresponding non-oriented edge.
$p^1_{B,T,u_0}((x,y))=\frac{1}{n-1} \chi_T(xy) \mathbf{1}_{d_T(x,u_0)>d_T(y,u_0)}$ where $d_T$ is the graph distance on the tree $T$.
And the output, claimed by Alice given $U,(x,y)$, is 
$\omega_{A,U}(x,y)=(n-1)\mathbf{1}_{x\in U} \mathbf{1}_{y\notin U}$.

One can check that 
$\E_{A,B}[\omega | A\leftarrow U, B\leftarrow T] =S(U,T)$ for all $(U,T)$ (see Claim \ref{spanningprotocorrect} in Appendix \ref{section:appendixmarkov}) . And this yields
a (non-negative) factorization\footnote{Indeed, set $A_{U,(u,\vec{e})}= \frac{1}{n-1} p^0_{A,U}(u) \omega_{A,U}(\vec{e})=\frac{1}{|U|} \mathbf{1}_{u\in U} \mathbf{1}_{x\in U} \mathbf{1}_{y\notin U}$ if $\vec{e}=(x,y)$. \hspace{5mm} And set $B_{(u,\vec{e}), T}=(n-1) p^1_{B,T,u}((x,y))=\chi_T(xy) \mathbf{1}_{d_T(x,u)>d_T(y,u)}$.} of the slack matrix $S$, of size at most $2nm \leq n^2(n-1)$ (if $|E|=m$ and $|V|=n$).

This remark holds in greater generality. Given $f : \mathcal{X} \times \mathcal{Y}\to \R_{\geq 0}$ a matrix, and $\pi$ a Markovian protocol in $k$ rounds which is correct for $f$, in the sense of Definition \ref{dfncorrect}, one deduces a factorization $f=AB$ as follows. Say Alice is the first speaker, and Bob is the one claiming an output (so $k$ is odd). Then, for $(u_1, \ldots , u_k)\in V_1 \times \cdots \times V_k$, set :
\[A_{x,u_1, \ldots , u_k}:=p^0_{A,x}(u_1) p^{(2)}_{A,x}(u_2,u_3) \cdots p^{(k-1)}_{A,x}(u_{k-1},u_k) \hspace{2mm}, \text{and set } \]
\[B_{y,u_1, \ldots , u_k}:=p^{(1)}_{B,y}(u_1,u_2) p^{(3)}_{B,y}(u_3,u_4) \cdots p^{(k-2)}_{B,y}(u_{k-2},u_{k-1}) \omega_{B,y}(u_k) \hspace{2mm} ,\] 
or $B_{y,u_1, \ldots , u_k}:=p^{(1)}_{B,y}(u_1,u_2) \cdots p^{(k-2)}_{B,y}(u_{k-2},u_{k-1}) \E_B[\omega_B | B\leftarrow y, u_k]$, in case $\omega_B$ is non-deterministic. Then observe that \emph{correctness} of $\pi$ in the sense of Definition \ref{dfncorrect} is exactly saying that $(AB)_{x,y}=f(x,y)$.

Given $f : \mathcal{X} \times \mathcal{Y}\to \R_{\geq 0}$ a matrix with non-negative entries, consider $\text{Prot}_+(f)$, the set of all Markovian protocols (with either $A$ or $B$ as first speaker, and either $A$ or $B$ as a final claimer) which are \emph{correct} for $f$, in the sense of Definition \ref{dfncorrect} : for any entry, the average output matches $f(x,y)$ exactly, and $\PP_{A,B}(\omega \geq 0)=1$. Given a protocol $\pi\in \text{Prot}_+(f)$, we define $\Gamma(\pi)\subset V_1 \times \cdots \times V_k$, with $\gamma=(v_1, \cdots , v_k)\in \Gamma(\pi)$ if $\exists x\in \mathcal{X}, \exists y\in \mathcal{Y} : A_{x,\gamma}B_{\gamma,y}>0$. In words, $\gamma\in V_1\times \cdots \times V_k$ is in $\Gamma(\pi)$, if $\PP_{A,B}(\gamma| x,y) >0$ and $\omega_{D,z}(\gamma)>0$ for at least some $(x,y)\in \mathcal{X} \times \mathcal{Y}$ (with $(D,z)=(A,x)$ or $(B,y)$).

In the above example (with the Spanning Tree polytope), one has $\Gamma(\pi)=\{(u,(x,y))\in V\times \overrightarrow{E} : u\neq x, u \neq y\}$, because if Bob has received $u\in V$ from Alice, then the chances that he sends her some $(u,x) \in \overrightarrow{E}$, is zero (by definition of $p^1_{B,T,u}$), and although a path $\gamma=(u,(x,u))$, with $xu\in E$, has some non-zero chance of being the path of exchanged messages (i.e. $\max_{U,T} \PP_{A,B}(\gamma| U,T) >0$), such paths always end at a null claim ($\omega_{A,U}(x,u)=0$ if $u\in U$, and $p^0_{A,U}$ is such that $u\in U$ if $u$ is the first message sent). If $G=K_n$, this gives (since $u\neq x,y$ and $x\neq y$): $|\Gamma(\pi)|=n(n-1)(n-2)$ , hence the upper bound $\text{xc}(P_{SpT}(n)) \leq n(n-1)(n-2)$, recovering Martin’s bound \cite{martin91}.

\begin{prop}
\label{faenzarestated}
Let $S\in \R_{\geq 0}^{[m] \times [N]}$ be a non-negative matrix. Then: 
$$\text{rk}_+(S)=\min\{ |\Gamma(\pi)| : \pi \in \text{Prot}_+(f)\}.$$
\end{prop}

See \ref{faenzarestatedbis} in Appendix for a proof. As we are restricting ourselves to markovian protocols, Proposition \ref{faenzarestated} can be seen as a particular case of Theorem 2 in \cite{faenzaetal}. In fact non-markovian protocols can be made markovian (at the cost of blowing a width $w$ in to a width $w^k$, if $k$ is the number of rounds); also BPs can be articially used to support more general\footnote{i.e. not necessarily markovian, but still with a layered structure} randomized protocols
(see Remark \ref{remarkmarkovianornot}
in appendix for a further discussion comparing the models).

Note that if a (two-party) protocol is in one round of communication ($k=1$), then it is always markovian. In particular, the protocol we give in paragraph \ref{para:newprotpermu}, is markovian. However, the one we give for the matching polytope, is not, but it is \emph{almost markovian}, in the sense that the transition probabilities fall into the formalism we have just exposed, and in the sense that making it markovian would result in a width of $O(\ln(n)n^4 1.5^n)$ rather than the $O(\ln(n) n^3 1.5^n)$ which we find here.


\section{A protocol for the matching polytope}
\label{section:PMG}

%
In this section we describe a protocol, from which one can deduce an affine extension $Q_{match}(n)$ of size $O(\text{poly}(n)1.5^n)$, of the matching polytope.
In Edmonds' description of the matching polytope, there are 3 types of constraints: non-negativity constraints $x_e \geq 0$, vertex-constraints $x(\delta(v))=\sum_{e\ni v} x_e \leq 1$, and the constraints $x(E(U))\leq \frac{|U|-1}{2}$ for all odd subsets $U$ (such that $E(U)\neq \emptyset$). So the slack matrix has three horizontal blocks, and the main one (both in size and in non-negative rank) is the block corresponding to constraints $x(E(U))\leq \frac{|U|-1}{2}$. The slack of a matching $M$ with respect to the constraint  $x(E(U))\leq \frac{|U|-1}{2}$, is $S_{U,M}=\frac{|U|-1}{2}-|M\cap E(U)|$.

Assume that for each $k\in [|1, \lfloor |V|/2\rfloor |],$ some sets $T_k=\{X_1, \ldots , X_{n_k}\}\subset \binom{V_n}{k}$ have been found, with the property that for any matching $M$ of $G$, made of $k$ edges, there exists $X_j\in T_k$ such that $M\cap \delta(X_j)=M$ ($M$ and $X_j$ are \emph{compatible}).
One could simply take $T_k= \binom{V_n}{k}$, but Lovasz lemma (see \ref{lovaszineq} below) allows to take $T_k$ of size only $n_k \leq (1+k\ln(n))2^{-k}\binom{n}{k}$, where $n=|V|$, see Corollary \ref{lovaszcoro}  in Appendix \ref{section:appendixlovasz}, and see also Lemma 3 in the article \cite{faenzaetal}. We now describe the protocol : it is similar to the protocol given in \cite{faenzaetal} for the perfect matching polytope.

Say A received (as input) an odd subset $U$, of size $|U|\geq 3$, and B received a matching $M$ of size $|M|=k$. They both have access to the sets $T_k \subset \binom{V_n}{k}$, $T_k=\{X_1, \ldots , X_{n_k}\}$ with a given encoding.
\begin{itemize}
\item B finds $X_j\in T_k$, a $k$-subset of nodes such that $M\cap \delta(X_j)=M$. He sends $(k,j)$ to Alice.
\item A picks $Z\in \{X_j, X_j^c\}$ such that $|Z\cap U| < |Z^c \cap U|$. If $|Z\cap U|=0$, then she stops the protocol and claims $\frac{|U|-1}{2}$. Else, she picks $u\in Z\cap U$ uniformly at random, then sends $u$ to Bob.
\item B finds $u'=M(u)$, and he sends $u'$ to Alice.
\item If $u'\in U$, then A claims $\frac{|U|-1}{2}-|Z\cap U|$. Otherwise she claims $\frac{|U|-1}{2}$.
\end{itemize}

\begin{figure}[ht]
\centering
\includegraphics[width=0.7 \textwidth]{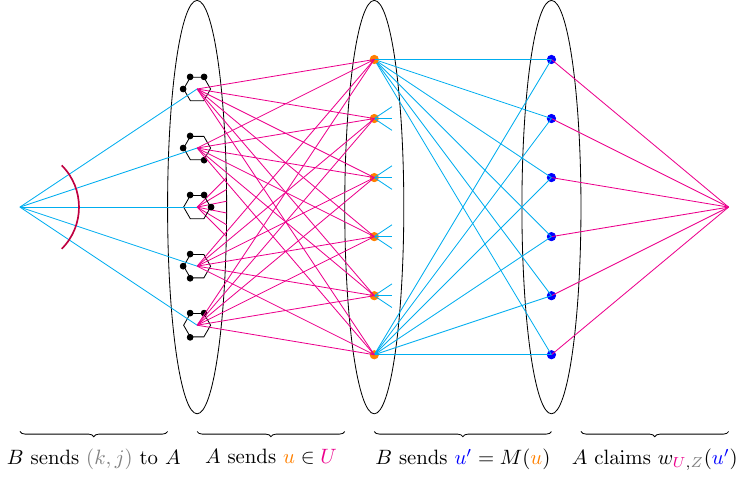}
\caption{An almost Markovian protocol for the matching polytope}
\end{figure}
Comments: upon receiving $(k,j)$ by the first round of communication, Alice learns that Bob has a matching of size $k$ at hand, and that he is pointing to the $k$-set $X=X_j$ in $j^{th}$ position in the list $T_k$ which is public (or at least shared by the two parties). In the second round, the point of choosing $Z=X_j$ or $Z=X_j^c$, rather than just always picking $Z=X_j$, is to ensure that the claimed output is always non-negative.

\footnote{Figure for $n=6$, $k=3$, one has $|T_3|=5$, given by the greedy algorithm, which here is in fact optimal.}Following \cite{faenzaetal}, if $\pi$ is a protocol which outputs a claim $\omega$, we shall say that $\pi$ is correct 
if $\PP_{A,B}[\omega \geq 0]=1$ and $\E_{A,B}[\omega| A \leftarrow U ; B\leftarrow M]=S_{U,M}$ for any pair $(U,M)$, i.e. if Alice and Bob always produce a non-negative claim, and that on average\footnote{the only source of randomness here is the private randomness of A and B used through the protocol ; the inputs are deterministic.} their claim is the slack, for any inputs.



To be more precise, Alice should be the first speaker of the protocol : she will send $e$ or  $v$, if the constraint she was given is not some $x(E(U))\leq \frac{|U|-1}{2}$ (i.e. not some odd $|U| \geq 3$), but rather an $x_e\geq 0$ or an $x(\delta(v))\leq 1$ constraint. Upon receiving $e\in E$, Bob can claim $1_{e\in M}$, which is the slack. Upon receiving $v\in V$, he can claim\footnote{(we denote $V(M)=\cup_{e\in M} e$ is the set of vertices $v\in V$ touched by $M$)} $1-\chi_M(\delta(v))=1_{v\notin V(M)}$ , which is the slack. And if she received for input some $U$, she sends $\emptyset$ to indicate to Bob that they shall follow the above protocol (i.e. he picks $X\in T_k$, etc.). So formally, we should add one round of communication (Alice speaks first), at the beginning. 

\subsubsection{Correctness of the protocol given for the matching polytope}
\label{para:matchingproto}

It is easily seen that the output is always non-negative (thanks to Alice picking $Z\in \{X_j,X_j^c\}$, accordingly), it remains to justify why, for any matching $M$ (of size $k\geq 1$), and any odd subset $U$ (of size $|U| \geq 3$), one has
\[\E_{A,B}(\omega_A | A\leftarrow U, B\leftarrow M)=\frac{|U|-1}{2}-|M\cap E(U)|=S(U,M).\]

Say B chooses $X_j\in T_K$ according to some probability distribution $p^0_M\in \mathcal{P}(T_k)$, such that $\text{supp}(p^0_M) \subset \{X\in \binom{V_n}{k} : \delta(X)\cap M=M\}$.

\begin{proof}[Proof of correctness]
If $|Z\cap U|=0$, meaning $U\subset X_j$ or $U\cap X_j=\emptyset$, then, since $X_j$ was chosen so that $M\subset E(X_j, X_j^c)$, Alice knows for a fact that $M \cap E(U)=\emptyset$ and hence that $S_{U,M}=\frac{|U|-1}{2}$. Hence her claim is exactly the slack in these cases.

Otherwise, she picks $u\in Z\cap U$ uniformly at random. Let us call $p=|E(U)\cap M|$, the number of edges of $M$ with both extremities in $U$, a number unknown to either Alice or Bob. Hence (by compatibility of $X_j$ with $M$) $p=|M\cap E(U\cap Z, U\cap Z^c)|$, i.e. $p$ is the  number of $u\in Z\cap U$ such that $M(u)\in Z^c\cap U$. Therefore, the average output claimed by Alice is
$$\E_A(\omega_A | U, Z)=\frac{|U|-1}{2} -q |Z\cap U|$$
where $q\in [0,1]$ is the chances that Alice chooses $u\in Z\cap U$ such that $M(u)\in Z^c \cap U$. So $q=p / |Z\cap U|$, and thus (no matter which $X_j$ Bob had picked, as long as it was compatible to his $M$), the average claim is $\E_A(\omega)=\frac{|U|-1}{2}-p=\frac{|U|-1}{2}-|E(U)\cap M|=S_{U,M}$.

More formally, what we have just justified is that for any distribution $p^0_M \in \mathcal{P}(T_k)$ used by Bob to pick $X_j$, as long as $\text{supp}(p^0_M)\subset \{X \in {V\choose n} : \delta(X) \cap M=M\}$, one has $\E_A(\omega_A| A\leftarrow U, X_j)=S_{U,M}$ which implies that
\begin{align*}\E_{A,B}(\omega_A | A\leftarrow U, B\leftarrow M) &=\sum_{X_j\in \text{supp}(p^0_M)} p^0_M(X_j) \E_A(\omega_A | A\leftarrow U, X_j) \\
&=\sum_{X_j : \delta(X) \cap M=M} p^0_M(X_j) \E_A(\omega_A | A\leftarrow U, Z)=S_{U,M}
\end{align*}
(the second last equality is because $Z$ is determined by $(U, X_j)$). Since we have also seen that $\PP_{A,B}(\omega_A\geq 0) =1$, we see that the protocol is correct.
\end{proof}

Now, following Proposition \ref{faenzarestated}, one can check that the above protocol yields a factorization $S=AB$ : we leave the details in paragraph \ref{para:protomatchgivesfacto}, in Appendix \ref{section:appendixlovasz}.
Therefore (we recall that $n_k \leq (1+k\ln(n))2^{-k}\binom{n}{k}$) :
\[\text{rk}_+(S)\leq |\Gamma(\pi)|\leq |E|+|V|+ n(n-1)\sum_{k=1}^{\lfloor n/2 \rfloor} n_k \leq n+\frac{n(n-1)}{2} (1+2\sum_k n_k) \leq n^3\ln(n) 1.5^n\]
and hence $\text{xc}(P_{\text{match}}(n))\leq n^3\ln(n) 1.5^n$, by Yannakakis' theorem.

\section{Sorting networks give protocols for the permutahedron}
\label{section:permu}

Edmonds representation of the permutahedron (see \eqref{edmondsreppermu}) yields the slack matrix $S$
with entries $S_{J,\sigma}=\sigma(J)-\frac{|J|(|J|+1)}{2}=\langle x_{\sigma} ,\chi_J\rangle- \langle x_{id}, \chi_{\{1,\cdots, |J|\}}\rangle $,
(thus with $2^n-2$ rows and $n!$ columns).



The permutahedron has dimension $n-1$, it has $n!$ vertices and $2^n-2$ facets 
We have seen that the Birkhoff polytope is an $n^2$-size lift of $\text{Perm}(n)$.
%
%
 Goemans (\cite{goemans}) has constructed an affine extension $Q_n$ of $Perm(n)$ even smaller.
His construction relies on sorting networks. In fact, he shows that given a sorting network $\sigma^+$ in $q$ comparators, one can use them\footnote{see Appendix \ref{para:comparisongoemans} , where we recall Goemans' construction, so as to compare it with the extension yielded by our protocol} to construct $Q_n \subset \mathbb{R}^{n+2q}$, described by only $2q$ inequalities, so that $\text{xc}(\text{Perm}(n))\leq |Q_n| \leq 2q$.
A result by Ajtai, Komlos, Szemeredi, says that there exist SNs with only $q \leq C n \ln n$ comparators, so that Goemans' construction is optimal (within a multiplicative factor), since any extension of $\text{Perm}(n)$ must have at least $n!$ vertices, thus at least $\log_2(n!)=\Omega(n\ln(n))$ facets. Let us briefly recall the definition of sorting networks, before giving a randomized protocol which results in an affine extension $Q'_n$ of $\text{Perm}(n)$, also of size $O(q)$ and dimension $O(q)$.


Given $(i,j)\in [n]^2$ a pair of coordinates, with $i<j$, set $\theta_{i,j}=\frac{1}{\sqrt{2}} (e_j-e_i)$, where $(e_{\ell})_{1\leq \ell \leq n}$ is the canonical basis of $\R^n$. Denote $H_{i,j}^+=\{x : x_i \leq x_j\}=\{x: \langle \theta_{i,j},x\rangle \geq 0\}$, and denote $\sigma_{i,j}^+$ the associated conditional reflection,
i.e. $\sigma_{i,j}^+(x)=\pi_{i,j}(x)+|\langle \theta_{i,j},x\rangle|\theta_{i,j}$, where $\pi_{i,j}$ denotes the (orthogonal) projection onto $H_{i,j}=\theta_{i,j}^{\perp}$.
For instance, recall that $x_{\sigma}=(\sigma(1), \ldots ,\sigma(n))$, if $\sigma \in\mathcal{S}_n$. So $\sigma_{i,j}^+(x_{\sigma})=x_{\sigma}$ if $\sigma(i)<\sigma(j)$, and $\sigma_{i,j}^+(x_{\sigma})=\tau_{i,j} \cdot x_{\sigma}=x_{\sigma \circ \tau_{i,j}}$ otherwise\footnote{$\mathcal{S}_n$ acts on $\R^n$ by permuting coordinates. For $\tau\in \mathcal{S}_n$ and $x\in \R^n$, one denotes $\tau \cdot x=(x_{\tau^{-1}(k)})_{k\in [n]}$. In particular, $\tau \cdot x_{\sigma}=x_{\sigma \circ \tau^{-1}}$. If $\tau=\tau_{i,j}$ is the transposition permuting $i$ and $j$, then $\tau^{-1}=\tau$ and so $\tau \cdot x_{\sigma}=x_{\sigma \circ \tau}$.}. By extension, for $\sigma\in \mathcal{S}_n$, we also write $\sigma_{i,j}^+(\sigma):=\sigma$ if $\sigma(i)<\sigma(j)$, and $\sigma_{i,j}^+(\sigma):=\sigma \circ \tau_{i,j}$ if $\sigma(i)>\sigma(j)$.

Assume we are given (at least, Alice and Bob shall be given) a sorting network $\mathcal{C}=((i_{\ell},j_{\ell}))_{\ell=0}^{q-1}$, that is a sequence of $q$ pairs $i_{\ell}<j_{\ell}$ such that $\sigma^+ := \sigma_{i_{q-1},j_{q-1}}^+ \circ \cdots \circ \sigma_{i_0,j_0}^+$ is sorting any $x_{\sigma}$ onto $x_{id}$ (see paragraph \ref{section:appendixsorting} for equivalent definitions).

In particular, given this sorting network $\sigma^+$, if $\sigma\in \mathcal{S}_n$, and $(\sigma^{(k)})_{0\leq k\leq q}$ is the sequence of permutations defined by $\sigma^{(0)}:=\sigma$, and $x_{\sigma^{(k+1)}}:=\sigma_{i_k,j_k}^+(x_{\sigma^{(k)}})$, for $0\leq k \leq q-1$, then we always have $\sigma^{(q)}=id_{[n]}$. We will abuse notations slightly and sometimes write $\sigma^{(k+1)} :=\sigma_{i_k,j_k}^+(\sigma^{(k)})$.

From the same sequence of pairs $\mathcal{C}$, one can define the operator $\sigma^-=\circ_{l=0}^{q-1} \sigma_{i,j}^-$ (with $\sigma_{i,j}^-$ the condtional reflection onto $H_{i,j}^-$), and one easily checks (see Lemma \ref{forwardisreverseandversa} in Appendix \ref{section:appendixsorting}) that $\sigma^+$ defines a SN, if and only if $\sigma^-$ defines a \emph{reverse} SN, meaning that  $\sigma^-(x_{\sigma})=\tau_0 \cdot x_{id}=(n, \ldots ,1)$ for all $\sigma\in \mathcal{S}_n$.
We thus have\footnote{where we similarly abuse notation, and write $\sigma_{i,j}^-(J)=K$ to mean that $J, K\in \mathcal{P}([n])$ are such that $\sigma_{i,j}^-(\chi_J)=\chi_K$} $\sigma^-(J)=[|J|]$ for any non-empty, proper subset $J\subset [n]$. Which means that for any $J\subset [n]$, if $(J^{(k)})_{0\leq k\leq q}$ is the sequence of subsets of $[n]$ (all of size $|J|$) given by $J^{(0)}=J$ and then $J^{(k+1)}=\sigma_{i_k,j_k}^-(J^{(k)})$ for $0\leq k\leq q-1$, then $J^{(q)}=\{1,2, \ldots , |J|\}=[|J|]$, for any (non-empty) $J\subset [n]$.

Given a sorting network $\sigma^+$ as above (to which at least A and B have access), say Alice receives $J$ for input, and Bob receives $\sigma$ for input. Then Alice computes privately the sequence $(J_k)_{0\leq k\leq q}$ (with $J_{k+1}=\sigma_{i_k,j_k}^-(J_k)$), and Bob computes privately the sequence of permutations $(\sigma^{(k)})_{0\leq k\leq q}$ (with $\sigma^{(k+1)} :=\sigma_{i_k,j_k}^+(\sigma^{(k)})$). Observe than the slack is given by
\begin{equation}
\label{slackpermu}
S_{J,\sigma}=\sigma^{(0)}(J_{0})-\sigma^{(q)}(J_q)=\sum_{l=0}^{q-1} \sigma^{(l)}(J_l) - \sigma^{(l+1)}(J_{l+1})=:\sum_{l=0}^{q-1} \delta_{i_l,j_l}(\sigma^{(l)}, J_l)
\end{equation}
by setting $\delta_{i,j}(\sigma, J)=\sigma(J)-\sigma_{i,j}^+(\sigma)(\sigma_{i,j}^-(J))$.
One can check (see Claim \ref{neededclaimsorting} in Appendix \ref{section:appendixsorting}) that for any $\sigma$, and any $J$, one has (where, if $a\in \R$, one denotes  $a_+=\max \{0,a\}$):
\begin{equation}
\label{eqn:deltarewritten}
\delta_{i,j}(\sigma, J)=\mathbf{1}_{j\in J, i\notin J} (\sigma(j)-\sigma(i))_+ +\mathbf{1}_{i\in J, j\notin J} (\sigma(i)-\sigma(j))_+
\end{equation}

Given $J\subset [n]$, and a sequence $(J_\ell)_{\ell \in [|0,q-1|]}$ (with $J_0=J$ ; $J_{\ell+1}=\sigma_{i_{\ell},j_{\ell}}^-(J_{\ell})$), define $\epsilon(l)=\epsilon(l, J)\in \{-1,0,1\}$ by :
$\epsilon(l)=1$ if [$j_l\in J_l$ and $i_l \notin J_l$] ; $\epsilon(l)=-1$ if [$j_l\notin J_l$ and $i_l\in J_l$] ; and $\epsilon(l)=0$ otherwise. Therefore $J$ (and the sequence $(J_l)$ induced by $J$ and by $\sigma^-$) gives rise to a partition of $[|0,q-1|]$, namely set $A_x=\{l\in [|0,q-1|] : \epsilon(l,J)=x\}$, for $x=-1,0,1$. Hence, \eqref{slackpermu} and \eqref{eqn:deltarewritten} give the following rewriting of the slacks:
\begin{equation}
\label{slackpermutahedronformula}
S_{J,\sigma}=\sum_{l=0}^{q-1} \delta_{i_l,j_l}(\sigma^{(l)}, J_l)=\sum_{l\in A_1} (\sigma_l(j_l)-\sigma_l(i_l))_+ +\sum_{l\in A_{-1}} (\sigma_l(i_l)-\sigma_l(j_l))_+ .
\end{equation}
From equation \eqref{slackpermutahedronformula} one can deduce the following two-round protocol which, for any given inputs $(J,\sigma)$ will, on average, output the slack $S_{J,\sigma}$. (so Alice receives $J$ as input, and Bob receives $\sigma$.)
\begin{itemize}
\item Using $\sigma_{\mathcal{C}}^-$, Alice determines the sequence $(J_l)_{0\leq l \leq q}$, and deduces a partition $[|0,q-1|]=A_0 \cup A_1 \cup A_{-1}$ (as explained above, via colors $\epsilon(l,J)\in\{0,\pm 1\}$). If $|A_0|=q$, then she claims $0$. Else she picks $l\in A_1 \cup A_{-1}$, uniformly at random. Denote $k:=q-|A_0| \geq 1$. She sends $l$ to Bob.
\item Using $\sigma^+_{\mathcal{C}}$, Bob computes the sequence of permutations $(\sigma_l)_{0\leq l \leq q}$ (with $\sigma_0=\sigma$, and $\sigma_q=id_{[n]}$). Receiving $l\in [|0,q-1|]$, he sends $w_B:=\sigma_l(j_l)-\sigma_l(i_l)$ to Alice.
\item If $l\in A_1$ then she claims $(k\cdot w_B)_+=\max\{0, k\cdot(w_B)\}$, and if $l\in A_{-1}$, then she claims $(k\cdot w_B)_-=\max\{0, -kw_B\}$.
\end{itemize}

It is clear by \eqref{slackpermutahedronformula}, that the average claim made by Alice, if A and B follow the above protocol, is $S_{J,\sigma}$. Moreover, her claim (the output) is non-negative with probability one. In this sense, the above protocol is \emph{correct}, meaning it yields to a non-negative factorization $S=AB$ of the slack matrix. However, it yields a factorization of size $2qn \geq n^2\ln(n)$ (since any sorting network has size $q\geq \log_2(n!)=\Omega(n\ln(n))$), i.e. of size greater than the size of the Birkhoff polytope.
It turns out that to divide the width of the above protocol by a factor $\Omega(n)$ (corresponding to Bob communicating $w_B\in [|-n,n|]$), it suffices to wrap it into a one-round protocol, as we now explain.

\subsection{A one-round protocol for guessing slacks of the permutahedron}
\label{para:newprotpermu}

Assume we are given $(i_0, j_0), \dots, (i_{q-1}, j_{q-1})$, which defines $\sigma^+=\circ_{\ell=0}^{q-1} \sigma_{i_{\ell},j_{\ell}}^+$ a valid sorting network (see paragraph \ref{section:appendixsorting}). Alice and Bob, after receiving as an input, respectively, $J$ for A and $\sigma$ for B, will proceed as follows:

\begin{itemize}
    \item Alice computes the sequence $(J_l)_{0 \leq l \leq q}$ according to $\sigma^-$.
    i.e., $J_0 := J$; $J_{k+1} = \sigma_{i_k, j_k}^-(J_k)$ for all $k=0,1,\ldots,q-1$.
    As $\sigma^-$ is a reverse SN, we always have $J_q = \{1, \dots, |J|\}$.
    While performing this sequence of transformations\footnote{we note that she doesnt' need to keep in memory the sequence $(J_k)$, she only needs to have the current $J_k$ (so as to compute the next), however she needs to store $\epsilon\in \{-1,0,1\}^q$, so spacewise she needs $O(n+q)$, and she can perfom finding $\epsilon$ within time $O(q)$.}, she marks the indices $l \in [|0, q-1|]$ with a "color" $\epsilon \in \{-1, 0, 1\}$ as she goes, with 
    $\epsilon(l):=\mathbf{1}_{j_l\in J_l, i_l\notin J_l} - \mathbf{1}_{i_l\in J_l, j_l\notin J_l}$.
She ends up with a partition of $[|0,q-1|]=A_0\cup A_{-1}\cup A_1$, where $A_x=\{ l : \epsilon(l)=x\}$.

    \item Alice chooses $l \in [|0, q-1|]$ uniformly at random (probability $1/q$ for each $l$). Say $l\in A_{\epsilon}$, with $\epsilon \in \{-1,0,1\}$.
    She sends $(l, \epsilon)$ to Bob.

    \item Having received $(l, \varepsilon)$, Bob computes the permutation $\sigma_l$, from his input $\sigma_0 := \sigma$ (and using $\sigma^+$). In other words, he computes $x_{\sigma_l} := \sigma_{i_{l-1}, j_{l-1}}^+ \circ \dots \circ \sigma_{i_0, j_0}^+(x_\sigma)$, and then he computes the value $w=w_l(\sigma)= \sigma_l(j_l) - \sigma_l(i_l)$ and he claims (as output of the protocol) $(q \varepsilon w)_+ = \max\{0, q \varepsilon w\}$.
\end{itemize}

Given a pair of inputs $(J, \sigma)$, observing that the claim $w_B=(q \varepsilon w_l(\sigma))_+$ is always $0$ when $\varepsilon=0$, we see that the average value of the output produced by this protocol is (the second inequality is by \eqref{slackpermutahedronformula}) :
\begin{align*}
\E_A( w_B | A \leftarrow J; B \leftarrow \sigma ) = \sum_{l \in A_1} (\sigma_l(j_l) - \sigma_l(i_l))_+ + \sum_{l \in A_{-1}} (\sigma_l(i_l) - \sigma_l(j_l))_+ = S_{J,\sigma} 
\end{align*}
Therefore, it gives a non-negative factorization $S=AB$ of size $2q$.
Set $A_{J,l,\varepsilon}=1_{\epsilon(l,J)=\varepsilon}$ and $B_{l,\varepsilon,\sigma}=(\varepsilon w_l(\sigma))_+$
where $\epsilon(l,J)\in \{-1,0,1\}$ is the color Alice gave to $l$ while computing (according to $\sigma^-$) her sequence $(J_l)$, starting from $J_0:=J$, and where $w_l(\sigma)=\sigma_l(j_l)-\sigma_l(i_l)$, with $\sigma_l=\sigma_{i_{l-1},j_{l-1}}^+ \circ \cdots \circ \sigma_{i_0,j_0}^+(\sigma)$, i.e. $\sigma_l$ the $l^{th}$ permutation computed by B, starting from $\sigma_0=\sigma$ and following $\sigma^+$. We can drop all $q$ indices $(l,0)$ because the corresponding rows in $B$ are identically $0$, so that the factorization has size (at most) $2q$.

We conjecture that if $\sigma^+$ is a minimal sorting network of $\R^n$ (see paragraph \ref{para:minimalSN} for a definition of minimality in this context), in $q$ comparators, then any randomized protocol with inputs in $\mathcal{P}([n])\times\mathcal{S}_n$, and whose average output is $S(J,\sigma)$ (for any given pair of inputs), must have width at least $q$.

\printbibliography

@inproceedings{faenzaetal,
  author       = {Yuri Faenza and       Samuel Fiorini and
                  Roland Grappe and
                  Hans Raj Tiwary},
  editor       = {Ali Ridha Mahjoub and
                  Vangelis Markakis and
                  Ioannis Milis and
                  Vangelis Th. Paschos},
  title        = {Extended Formulations, Nonnegative Factorizations, and Randomized
                  Communication Protocols},
  booktitle    = {Combinatorial Optimization - Second International Symposium, {ISCO}
                  2012, Athens, Greece, April 19-21, 2012, Revised Selected Papers},
  series       = {Lecture Notes in Computer Science},
  volume       = {7422},
  pages        = {129--140},
  publisher    = {Springer},
  year         = {2012},
  url          = {https://doi.org/10.1007/978-3-642-32147-4\_13},
  doi          = {10.1007/978-3-642-32147-4\_13},
  bibsource    = {dblp computer science bibliography, https://dblp.org}
}

@article{fiorinibirkhoff,
title = {Combinatorial bounds on nonnegative rank and extended formulations},
journal = {Discrete Mathematics},
volume = {313},
number = {1},
pages = {67-83},
year = {2013},
issn = {0012-365X},
doi = {https://doi.org/10.1016/j.disc.2012.09.015},
url = {https://www.sciencedirect.com/science/article/pii/S0012365X12004207},
author = {Samuel Fiorini and Volker Kaibel and Kanstantsin Pashkovich and Dirk Oliver Theis},
}

@article{Edmonds65matchings,
  title={Maximum matching and a polyhedron with 0,1-vertices},
  author={Jack Edmonds},
  journal={Journal of Research of the National Bureau of Standards Section B Mathematics and Mathematical Physics},
  year={1965},
  pages={125},
  url={https://api.semanticscholar.org/CorpusID:15379135}
}

@article {FRT12,
    AUTHOR = {Fiorini, Samuel and Rothvo\ss , Thomas and Tiwary, Hans Raj},
     TITLE = {Extended formulations for polygons},
   JOURNAL = {Discrete Comput. Geom.},
  FJOURNAL = {Discrete \& Computational Geometry. An International Journal
              of Mathematics and Computer Science},
    VOLUME = {48},
      YEAR = {2012},
    NUMBER = {3},
     PAGES = {658--668},
      ISSN = {0179-5376},
   MRCLASS = {52B05 (68U05 90C57)},
  MRNUMBER = {2957636},
MRREVIEWER = {Tamon Stephen},
       DOI = {10.1007/s00454-012-9421-9},
}

@article {goemans,
    AUTHOR = {Goemans, Michel X.},
     TITLE = {Smallest compact formulation for the permutahedron},
   JOURNAL = {Math. Program.},
  FJOURNAL = {Mathematical Programming},
    VOLUME = {153},
      YEAR = {2015},
    NUMBER = {1, Ser. B},
     PAGES = {5--11},
      ISSN = {0025-5610},
   MRCLASS = {90C10 (90C57)},
  MRNUMBER = {3395538},
       DOI = {10.1007/s10107-014-0757-1},
}

@article{lovasz,
title = {On the ratio of optimal integral and fractional covers},
journal = {Discrete Mathematics},
volume = {13},
number = {4},
pages = {383-390},
year = {1975},
issn = {0012-365X},
doi = {https://doi.org/10.1016/0012-365X(75)90058-8},
url = {https://www.sciencedirect.com/science/article/pii/0012365X75900588},
author = {L. Lovász}
}

@article{kaibelw2015,
  title     = {A Short Proof that the Extension Complexity of the Correlation Polytope Grows Exponentially},
  author    = {Kaibel, Volker and Weltge, Stefan},
  journal   = {Discrete \& Computational Geometry},
  volume    = {53},
  number    = {2},
  pages     = {397--401},
  year      = {2015},
  publisher = {Springer},
  doi       = {10.1007/s00454-014-9638-7}
}

@article{martin91,
title = {Using separation algorithms to generate mixed integer model reformulations},
journal = {Operations Research Letters},
volume = {10},
number = {3},
pages = {119-128},
year = {1991},
issn = {0167-6377},
doi = {https://doi.org/10.1016/0167-6377(91)90028-N},
url = {https://www.sciencedirect.com/science/article/pii/016763779190028N},
author = {R.Kipp Martin},
}

@phdthesis{szusterman,
  author       = {Maud Szusterman}, 
  title        = {Contributions to algebraic complexity, extension complexity, and affine convex geometry
},
  school       = {Université Paris Cité},
  year         = 2023,
  url          = {https://theses.fr/api/v1/document/2023UNIP7355},
  id           = {HAL/tel-04902817}
}

@article{yannak91,
title = {Expressing combinatorial optimization problems by Linear Programs},
journal = {Journal of Computer and System Sciences},
volume = {43},
number = {3},
pages = {441-466},
year = {1991},
issn = {0022-0000},
doi = {https://doi.org/10.1016/0022-0000(91)90024-Y},
url = {https://www.sciencedirect.com/science/article/pii/002200009190024Y},
author = {Mihalis Yannakakis}
}


\section{Appendix}

\subsection{Appendix on Markovian protocols}
\label{section:appendixmarkov}

We quoted the protocol (of width $|\Gamma(\pi)|\leq n(n-1)(n-2)$) from \cite{faenzaetal}, for the spanning tree polytope:

\begin{itemize}
\item Alice picks $u_0\in U$ uniformly at random. She sends $u_0$ to B.
\item Bob picks $e\sim \text{Unif}(T)$. He orients $e$ towards $u_0$ (on $T$), and sends $\vec{e}=(u,v) \in \overrightarrow{E}$ to A.
\item Alice claims $\omega_A=(n-1)$ if $u\in U$ and $v\notin U$, and claims $0$ otherwise ;
\end{itemize}
and we claimed that the choice of $p^0_{A,U}=\text{Unif}(U)$ as an initial distribution (which tells A how to pick her $u_0$ to be sent to B in the first round) wasn't crucial. The following claim states this more formally, and explains why any $p^0_{A,U}$ satisfying $\text{supp}(p^0_{A,U})\subset U$, will do the job.
\begin{claim}
\label{spanningprotocorrect}
For any spanning tree $T$, any subset $U$ (of size $|U|\geq 2$), and any $u_0\in U$:
\[ \E_B(\omega_A | A\leftarrow U , B\leftarrow T, u_0)=|U|-1-|T\cap E(U)|=S(U,T)\]
i.e. no matter how $u_0\in U$ is picked by Alice (to be sent in the first round), the average output claimed by Alice in the end is $S(U,T)$, if Bob picks $\vec{e}$ according to the transition probability $p^1_{B,T,u_0}$.
\end{claim}
(we recall that these transition probabilities, from $V$ to $\overrightarrow{E}$, are described by $p^1_{B,T,u}((x,y))=\frac{1}{n-1} \mathbf{1}_{xy\in T} \mathbf{1}_{d_T(x,u)>d_T(y,u)}$ ; in words, B orients all edges of $T$ towards $u$, and then picks one of these $(n-1)$ (oriented) edges, uniformly at random).
\begin{proof}
Indeed, if $f(U,T)=0$, it means that $(U, T\cap E(U))=:(U, T_U)$ is a tree (spanning $U$), so any oriented edge $\overline{xy}$ with $xy\in T$, starting from some $x\in U$, and pointing (on $T$) towards a fixed $u_0\in U$, must arrive in $U$ (otherwise it would contradict connectedness of $(U,T_U)$). Hence, if $f(U,T)=0$ and $\text{supp}(p^0_{A,U}) \subset U$, then $\PP_{A,B}(\omega_A=0 | U,T)=1$, so that $\E_{A,B}(\omega| U,T)=0=f(U,T)$, in that case.

Else, if say $f(U,T)=k-1$ for some $k\geq 2$, it means that $(U, T\cap E(U))$ is a spanning forest with $k$ connected components. Call them $C_0, C_1, \ldots , C_{k-1} \subset U$. Fixing $u_0\in C_0$, and orienting all adges of $T$ towards $u_0$, one can see that exactly $k-1$ of the $n-1$ oriented edges of $T$ (rooted at $u_0$), are of type $U\times U^c$ : each $C_j$ (for $j>0$) gives an edge $\overline{x_j y_j}$ which leaves $C_j$ (with $x_j\in C_j \subset U$, which minimizes $d_T(x,u_0)$, $x\in C_j$), and one easily checks that these are the only ones. Hence the ($B$)-probability that Alice claims $(n-1)$, is $q:=\frac{k-1}{n-1}$, so that again, $\E_{A,B}(\omega| U,T)=f(U,T)$.
\end{proof}

Hence, if $p^0_{A,U}$ is some distribution on $V$, with support in $U$, one deduces:
\begin{align*}
 \E_{A,B}(\omega_A| A\leftarrow U , B\leftarrow T)&=\sum_{u_0\in U} p^0_{A,U}(u_0) 
\E_B(\omega_A | A\leftarrow U , B\leftarrow T, u_0) \\
&=\sum_{u_0\in U} p^0_{A,U}(u_0) (|U|-1-|T\cap E(U)|)\\
&=|U|-1-|T\cap E(U)|=S(U,T).
\end{align*}

This example of a protocol, which, because it correctly \emph{guesses} the value $S(U,T)$, yields a non-negative factorization of $S$, holds in greater generality : one has $\text{rk}_+(S) \leq |\Gamma(\pi)|$ for any protocol $\pi$ such that ($\PP_{A,B}(\omega \geq 0| A\leftarrow i,B\leftarrow j)=1$ for all $i,j$, and  $\E_{A,B}(\omega | A\leftarrow i,B\leftarrow j)=S_{i,j}$, for all $(i,j)$. And moreover, any tight non-negative factorization of $S$, yields a protocol $\pi$ which is \emph{correct}, and such that $|\Gamma(\pi)|=\text{rk}_+(S)$. These two facts are gathered in the following proposition.

\begin{prop}
\label{faenzarestatedbis}
Let $S\in \R_{\geq 0}^{[m] \times [N]}$ be a non-negative matrix. Then: 
$$\text{rk}_+(S)=\min\{ |\Gamma(\pi)| : \pi \in \text{Prot}_+(f)\}.$$
\end{prop}
\begin{proof}
We have essentially already observed that if a protocol $\pi$, aimed at helping parties A, B guess a value $S(i,j)$, with $i\in [m]$ and $j\in [N]$, with $\pi$ defined on $V_1\times \cdots V_k$, with initial distribution $p^0_{A,x}\in \mathcal{P}(V_1)$, with transition probabilities $p^{(j)}_{D_j,z_j}(u_j, u_{j+1})$ (on $V_j \times V_{j+1}$), with $(D_j,z_j)=(A,x)$ or $(B,y)$ depending on parity of $j$), and with output $\omega_D=\omega_{D,z}(u_k)$, then
\[\E_{A,B}[\omega_D | A\leftarrow i ; B\leftarrow j ] =\sum_{u_1, \ldots , u_k} p^0_{A,x}(u_1) \prod_{j=1}^{k-1} p^{(j)}_{D_j,z_j}(u_j,u_{j+1}) \omega_{D_k,z_k}(u_k)  \] 
\[
= \begin{cases}
   \displaystyle \sum_{u_1, \dots , u_k} \left[p^0_{A,x}(u_1) \prod_{j=1}^{\lfloor (k-1)/2 \rfloor} p^{(2j)}_{A,x}(u_{2j},u_{2j+1}) \right] \cdot \left[ \left(\prod_{j=1}^{\lfloor k/2 \rfloor} p^{(2j-1)}_{B,y}(u_{2j-1},u_{2j})\right) \omega_{B,y}(u_k) \right] & \text{if $k$ odd} \\[2em]
   \displaystyle \sum_{u_1, \dots , u_k} \left[p^0_{A,x}(u_1) \left(\prod_{j=1}^{\lfloor (k-1)/2 \rfloor} p^{(2j)}_{A,x}(u_{2j},u_{2j+1}) \right)\omega_{A,x}(u_k) \right] \cdot \left[ \prod_{j=1}^{\lfloor k/2 \rfloor} p^{(2j-1)}_{B,y}(u_{2j-1},u_{2j}) \right] & \text{if $k$ even}
\end{cases}
\]

\[=:  \sum_{\gamma=(u_1, \ldots , u_k)} A_{i,\gamma} \cdot B_{\gamma, j} =\sum_{\gamma=(u_1, \ldots , u_k)\in \Gamma(\pi)} A_{i,\gamma} \cdot B_{\gamma, j}. \]
where the last equality is by definition of $\Gamma(\pi)$ : we don’t need to keep those $\gamma$ for which $A_{i,\gamma}=0$ for all $i$, or for which $B_{\gamma,j}=0$ or all $j$.

This shows that $\text{rk}_+(S)\leq \min\{|\Gamma(\pi)| : \pi \in \text{Prot}_+(S) \}$.

Now assume that $\text{rk}_+(S)=r$, hence there exists a non-negative factorization $S=AB$, with $A\in \R_{\geq 0}^{[m]\times [r]}$ and $B\in \R_{\geq 0}^{[r] \times [N]}$.
Up to multiplying $A$ by some $\lambda >0$ (and $B$ by $\lambda^{-1}$), assume that $\max_{i\in [m]} \sum_{k=1}^r A_{i,k}=1$.
Define $p^0_{A,i}$ to be a probability distribution on $[|0,r|]$, with $p^0_{A,i}(k):=A_{i,k} \in [0,1]$ if $k\in [r]$, and with $p^0_{A,i}(0):=1-\sum_{k=1}^r A_{i,k}$.
And set $\omega_{B,j}(0):=0$, and $\omega_{B,j}(k)=B_{k,j} \geq 0$. Then $\PP_{A,B}(\omega_B \geq 0)=1$, and
\[ \E_{A,B}[\omega_B | A\leftarrow i,B\leftarrow j]=\E_A[\omega_B | i,j ]=\sum_{k=0}^r p^0_{A,i}(k) \omega_{B,j}(k)=\sum_{k=1}^r A_{i,k} B_{k,j}=S_{i,j} \]
so that the protocol is correct (it is non-negative, and outputs the correct value, on average). Moreover here $|\Gamma(\pi)|\leq k$ (a priori $\Gamma(\pi)\subset [|0,k|]$, but $0\notin \Gamma(\pi)$, because $\omega_{B,j}(0)=0$ for all $j\in [N]$.)
\end{proof}

\emph{Remark :}
\label{remarkmarkovianornot}
     The above characterization \ref{faenzarestated} of $\text{xc}(P)$, or rather, of $\text{rk}_+(S)$, can be seen as a particular case of Theorem 2 in \cite{faenzaetal}, where a broader class of protocols is considered. The branching program structure does allow to recover their result, but somewhat artificially. To encompass non-markovian protocols, we shall allow more general probability distributions (and not restrict to \emph{transition probabilities} as we have), e.g. if the protocol is in $k=3$ rounds, with $V_1\times V_2\times V_k$ as sets of possible messages to exchange (and say $A$ speaks first), then $p^0_{A,x}\in \mathcal{P}(V_1), p^{(1)}_{B,y,u_1}\in \mathcal{P}(V_2)$ (corresponding to transition probabilities $p^{(1)}_{B,y}(u_1,u_2)$ of above), and one shall define probability distributions $p^{(2)}_{A,x,u_1,u_2}\in \mathcal{P}(V_3)$, for each possible $(u_1,u_2)$, i.e. Alice can also use former messages (namely, her choice of $u_1$) for deciding which $u_3\in V_3$ to send to Bob, and not just make her choice depend on lastly received $u_2$ (and on her input $x$). Similarly, the output $\omega_{D,z}$ shall be allowed to depend on $(u_1,\ldots,u_k)$ entirely, and not just on $u_k$. We chose to restrict to markovian protocols\footnote{Calling $|\Gamma(\pi)| \approx 2^{\text{cost}(\pi)}$ the \emph{width} of a given protocol, note that an arbitrary protocol of width $w$, can be transformed into a markovian protocol of width at most $w^k$, if $k$ is the number of rounds of communication.} here, mainly to keep lighter notations, and because branching programs are more adapted to markovian processes (they can also be used for non-markovian ones, at the cost of allowing both parties to have a $O(|V_1| \cdots |V_k|)$-space memory at their disposal).

\subsection{Lovasz analysis of the greedy algorithm for vertex cover, and consequences for matchings.}
\label{section:appendixlovasz}

A hypergraph $H=(V,E)$ is two finite sets $V$, $E$, together with an adjacency matrix $A\in \{0,1\}^{V\times E}$, with the $e^{th}$ column of $A$ indicating which $v\in V$ \emph{lie within} $e$. We call $H$ simple if no two columns of $A$ are identical (i.e. any (hyper)edge is characterized by the set of $v$ it contains) and no two rows of $A$ are identical (i.e. any $v\in V$ is characterized by $\delta(v):=\{e\in E | A_{v,e}=1\}$, i.e. by the set of edges which contain $v$).
A matching is a subset $M\subset E$ such that $||A M||_{\infty} \leq 1$, i.e. such that $\sum_{e\in M} A_{v,e}\leq 1$ for all $v\in V$ (where we abuse slightly notation and also denote $M$ its characteristic vector, $M\in \{0,1\}^E$). A (vertex) cover is a subset $W\subset V$ such that $\min_{e\in E} \sum_{v\in W} A_{v,e} \geq 1$, i.e. it is a subset $W$ which \emph{covers} all edges. For $k\in \mathbb{Z}_{\geq 1}$, one can define a $k$-matching as an integral weight $(\omega_e)_{e\in E} \in \mathbb{Z}_{\geq 0}^E$ such that $||A \omega||_{\infty} \leq k$, i.e. such that $\sum_{e\in E} A_{v,e} \omega_e \leq k$ for each vertex $v$ (equivalently, setting $M:=\{e\in E : \omega_e \geq 1\}$, one may think of $(\omega_e)_{e\in E} \in \mathbb{Z}_{\geq 0}^E$ as a multiset, with groundset $M$, and where $e\in M$ has multiplicity $\omega_e$). A $k$-matching is called simple if $\omega_e \in \{0,1\}$ for all $e\in E$. Similarly, for $k\in \mathbb{Z}_{\geq 1}$, one can define a $k$-cover as an integral weight $(t_v)_{v\in V} \in \mathbb{Z}_{\geq 0}^V$ such that $\forall e\in E, \sum_{v\in V} A_{v,e} t_v \geq k$, i.e. a cover such that each edge is covered by at least $k$ vertices (each vertex is counted with its multiplicity $t_v$). A $k$-cover is simple if $t_v\in \{0,1\}$ for all $v\in V$.

Given $H=(V,E)$ a hypergraph, denote $\nu_k(H)$ the maximal size of a $k$-matching on $H$, and $\tau_k(H)$ the minimal size of a $k$ cover on $H$. Denote $\nu(H)$, resp. $\tau(H)$, the maximal size of a matching, resp. the minimal size of a vertex cover (i.e. $\nu(H)=\nu_1(H)$, and $\tau(H)=\tau_1(H)$). The linear relaxations of these combinatorial quantities are defined via
\begin{itemize}
\item[(for ] VertexCover) : $\text{Cov}_k(H)=\{(t_v)\in \mathbb{R}_{\geq 0}^V : \sum_{v\in V}t_v A_{v,e}  \geq k , \forall e\in E \}$ ;
\item[(for ] Matchings) : $\mathcal{M}_k(H)=\{(\omega_e)\in \mathbb{R}_{\geq 0}^E : \sum_{e\in E} A_{v,e} \omega_e \leq k, \forall v\in V\}$.
\end{itemize}
the admissible (fractional) $k$-matchings, respectively the admissible (fractional) $k$-covers on $H$. Obviously, one has $\text{Cov}_k(H)=k\cdot \text{Cov}_1(H)=:k\cdot \text{Cov}(H)$ and $\mathcal{M}_k(H)=k\cdot \mathcal{M}_1 (H)=: k\cdot \mathcal{M} (H)$. Denote $\nu_k^*(H):=\max \{ \sum_{e\in E} \omega_e | \omega \in \mathcal{M}_k(H)\}$ and $\tau_k^*(H)=\min \{\sum t_v | (t_v) \in \text{Cov}_k(H)\}$. Thus $\nu_k^*(H)=k\nu^*(H)$ and $\tau_k^*(H)=k\tau^*(H)$. Moreover, by weak duality\footnote{denote $L(\omega,t)=\sum_{e} \omega_e +\sum_v t_v -\sum_{v,e} t_v A_{v,e} \omega_e=\sum_e \omega_e +\sum_v t_v(1-\sum_e A_{v,e} \omega_e)=\sum_v t_v -\sum_e \omega_e (\sum_v t_v A_{v,e} -1)$. Note that $\sup_{\omega \geq 0} \inf_{t\geq 0} L(\omega,t)=\nu^*$ and that $\inf_{t\geq 0} \sup_{\omega\geq 0} L(\omega,t)=\tau^*$. Weak duality simply states that $\nu^*=\sup_{\omega \geq 0} \inf_{t\geq 0} L(\omega,t) \leq \inf_{t\geq 0}  \sup_{\omega \geq 0} L(\omega,t)=\tau^*$. Here strong duality (i.e. $\nu^*=\tau^*$) holds , since replacing (both) $\R_{\geq 0}^E$ with $[0,1]^E$ and $\R_{\geq 0}^V$ with $[0,1]^V$, leaves both the $\sup \inf$ and the $\inf \sup$ unchanged ; but we don't need it. }, one has $\nu^*(H)\leq \tau^*(H)$.

A simple observation is that if $\omega\in \mathbb{Z}_{\geq 0}^E$ is a $k$-matching, then $m\cdot \omega=(m\omega_e)_e$ is an $mk$-matching. Similarly, if $t=(t_v) \in \mathbb{Z}_{\geq 0}^V$ is a $k$-cover, then $m\cdot t$ is an $mk$-cover. Hence $k\cdot \nu(H) \leq \nu_k(H)$ and $k\cdot \tau(H)\geq \tau_k(H)$. Altogether the previous observations yield the following inequalities (see \cite{lovasz}, page 384):
\begin{equation}
\label{easyineq}
\nu(H)\leq \frac{\nu_k(H)}{k} \leq \frac{\nu^*_k(H)}{k}=\nu^*(H)\leq\tau^*(H)=\frac{\tau^*_k(H)}{k} \leq \frac{\tau_k(H)}{k} \leq \tau(H)
\end{equation}

\emph{Observation} : denote $\nu^s_k(H) \leq \nu_k(H)$ the maximal size of a simple $k$-matching on $H$. Similarly denote $\tau^s_k(H)\geq \tau_k(H)$ the minimal size of a simple $k$-cover (if such cover exists). Observe that if $H=(V,E)$ has maximal degree at most $\Delta$, meaning $\sum_{e\in E} A_{v,e} \leq \Delta$ for all $v\in V$, then $|E|=|E(H)|=\nu^s_{\Delta}(H)$ (since any simple matching $M$ satisifies $M\subset E$, and since $M=E$ is a $\Delta$-matching).

Given $H=(V,E)$ an hypergraph, and $v\in V$, one denotes $\text{deg}_H(v)=\sum_{e\in E} A_{v,e}$ the \emph{degree} of $v$ (number of edges which \emph{contain} $v$).

Now, given an hypergraph $H=(V,E)$, which we assume to be simple (i.e. adjacency matrix $A\in \{0,1\}^{V\times E}$ has no two identical columns, and no two identical rows either), an efficient way to find a cover, i.e. to find $W\subset V$ such that $\min_{e\in E} \sum_{v\in W} A_{v,e} \geq 1$, is the greedy algorithm. For the sake of its analysis, let us describe in detail how it proceeds. One starts with $H_0=(V_0,H_0)=: H$, and finds $v_0\in V_0$, such that $\text{deg}_{H_0}(v_0)=\Delta_0 := \max_{v\in V_0} \text{deg}_{H_0}(v)$ (thus a vertex $v_0$ which \emph{covers} as many edges $e$ as possible). Then one sets $V_1=V_0\setminus \{v_0\}$ and $E_1=E_0\setminus \delta(v_0)=E_0\setminus \{e\in E : A_{v_0, e}=1\}$, and $H_1:=(V_1, E_1)$ (so $A_1$ is obtained from $A_0:=A$ by erasing the $v_0^{th}$ row from it, and erasing all columns which had a $1$ in this row). And again, find some $v_1\in V_1$ such that $\text{deg}_{H_1}(v_1)$ is maximal, $\text{deg}_{H_1}(v_1)=\max_{v\in V_1} \sum_{e\in E_1} A_{v,e}$. One continues in this fashion, picking at each step a vertex which covers as many (not covered yet) edges as possible ; until there is no edge left, i.e. until one reaches $H_t=(V_t,E_t)$ with $E_t=\emptyset$.
Denote $T=\{v_0, \ldots , v_{t-1}\}$ the vertices which have been picked along the way : $T$ is a vertex cover of $H$ (and such a cover $T$ is said to be obtained greedily).

The following inequality is due to L. Lovasz, and can be found in \cite{lovasz}.

\begin{prop}
\label{lovaszineq}
Let $H=(V,E)$ be a hypergraph. We denote $\Delta$ its maximal degree (if $A$ is the adjacency matrix of $H$, then $\Delta=\max_{v\in V} \sum_{e\in E} A_{v,e}$). Denote $\tau^*$ the optimal cost of a fractional cover of $H$, i.e. $\tau^*=\inf \{\sum_{v\in V} t_v : (t_v)\in \text{Cov}(H)\}$. Denote $T\subset V$ a vertex cover obtained by the greedy algorithm. Then
$$|T|\leq \left(1+\frac{1}{2}+\cdots+\frac{1}{\Delta}\right) \tau^*.$$
\end{prop}

For the sake of completeness, we quote here the proof from \cite{lovasz} , using the notations of above.

\begin{proof}[Proof (Lovasz):]
Thus let $H=:H_0=(V_0,E_0)$, $H_1=(V_1, E_1)$ (obtained after removing $v_0$ from $V_0$), ... until $H_t=(V_t, \emptyset)$, be the successive hypergraphs obtained following the greedy algorithm.

If $0\leq j \leq t-1$, denote $d_j=\text{deg}_{H_j}(v_j)$ the number of \emph{new} edges covered thanks to picking $v_j$ within $T$ (for instance, $d_0=\Delta$). By construction, for all $j\geq 0$, one has $d_{j+1}=\text{deg}_{H_{j+1}}(v_{j+1})\leq \text{deg}_{H_j}(v_{j+1})\leq d_j=\max_{v\in V_j} \text{deg}_{H_j}(v)$.
If $1\leq k\leq \Delta$, denote $t_k$ the number of $v_j$, $0\leq j \leq t-1$, such that $|E_{j+1}|=|E_j|-k$, i.e. the number of vertices $v_j$ ($0\leq j<t$) which resulted in exactly $k$ new edges being covered. Hence $|T|=t=t_{\Delta}+t_{\Delta-1}+ \cdots +t_2+t_1$. And denote $s_k=\text{card}\{j\geq 0 : d_j\geq k\}$, for instance $s_{\Delta+1}=0$, $s_{\Delta}=t_{\Delta}$, $s_{\Delta-1}=t_{\Delta}+t_{\Delta-1}$, ... $s_{2}=t_{\Delta}+\cdots+t_2=t-t_1$. Thus note that $\Delta_j=\max_{v\in V_j} \text{deg}_{H_j}(v)\leq k-1$, whenever $j\geq s_k$.

By the \emph{observation} following (\ref{easyineq}), one has $\nu^s_{\Delta}(H)=|E_0|=|E(H)|=\Delta t_{\Delta}+ \cdots +2 t_2 +t_1$, and also $\nu^s_k(H_{s_{k+1}})=|E(H_{s_{k+1}})|=kt_k+ \cdots + 2t_2+t_1$, for all $1\leq k \leq \Delta -1$. Now by (\ref{easyineq}), one also has $\nu^s_{\Delta}(H)\leq \nu_{\Delta}(H)\leq \nu^*_{\Delta}(H)=\Delta \cdot \nu^*(H)=: \Delta \cdot \nu^*$; and\footnote{the inequality $\nu^s_k(H_{s_{k+1}})\leq \nu^s_k(H_0)$ holds by construction : if $M\subset E(H_s)$, then $M$ doesn't see any of the $v_j, 0\leq j <s$, since $E_s=E_0\setminus \left(\cup_{0\leq j \leq s-1} \delta(v_j)\right)$, thus a simple $k$-matching on $H_s$ is also a simple $k$-matching on $H_r$, $0\leq r < s$.} $\nu^s_k(H_{s_{k+1}})\leq \nu^s_k(H_0)\leq k \cdot \nu^*(H_0)=k\cdot \nu^*$ .
Altogether, with $\Delta\geq k\geq 1$ (recall $s_{\Delta+1}=0$), one finds the inequalities :
\begin{align*}
\Delta t_{\Delta} +(\Delta-1)t_{\Delta-1}+ \cdots + 2 t_2 +t_1 &\leq \hspace{2mm}\Delta \cdot \nu^*  & (J_{\Delta})\\
& ... \\
k t_k + (k-1) t_{k-1} + \cdots + t_1 &\leq \hspace{2mm}k \cdot \nu^*  & (J_k)\\
& \cdots \\
2 t_2+t_1& \leq \hspace{2mm}2 \nu^* & (J_2)\\
t_1 & \leq \hspace{2mm}\nu^* &(J_1)
\end{align*}
Now, mutilpying $(J_k)$ by $\frac{1}{k(k+1)}$ for $1\leq k <\Delta$, multiplying $(J_{\Delta})$ by $\frac{1}{\Delta}$, and adding these $\Delta$ inequalities, one finds an inequality $\sum_{k=1}^{\Delta} \alpha_k t_k \leq \alpha_0 \nu^*$, with $\alpha_{\Delta}=\Delta \cdot \frac{1}{\Delta}=1$, and for $1\leq k<\Delta$, $\alpha_k=k \left( \frac{1}{k(k+1)}+\frac{1}{(k+1)(k+2)}+ \cdots + \frac{1}{(\Delta-1)\Delta}+\frac{1}{\Delta}\right)=1$ as well, and with $\alpha_0=\sum_{k=1}^{\Delta-1} \frac{k}{k(k+1)}+ \Delta \cdot \frac{1}{\Delta}=\sum_{j=1}^{\Delta} \frac{1}{j}$. Hence one finds
$$|T|=t_{\Delta}+\cdots+t_2+t_2 \leq \left(1+\frac{1}{2}+\cdots +\frac{1}{\Delta}\right) \nu^* \leq \left(1+\frac{1}{2}+\cdots +\frac{1}{\Delta}\right) \tau^*$$
as claimed.
\end{proof}

Here is a corollary of Propostion \ref{lovaszineq}, concerning $k$-matchings on the perfect graph $K_n=(V_n, E_n)$ (say $V_n=[n]$ and $E_n=\{\{i,j\} : 1\leq i<j \leq n\}$). 
Let us mention that the case $n$ even and $k=n/2$ of the following corollary of Lovasz's inequality \ref{lovaszineq}, was derived in \cite{faenzaetal}, using the same argument.
\begin{cor}
\label{lovaszcoro}
Let $1\leq k\leq \lfloor n/2 \rfloor$. Denote $\mathcal{M}_k$ the sets of matchings $M\subset E_n$, of size $|M|=k$. Denote $\mathcal{P}_k(V_n)$ the set of $k$-subsets of $V_n$. Then there exists $T_k \subset \mathcal{P}_k(V_n)$, of size at most $(1+k\ln(n))2^{-k}{n\choose k}$ and such that for any matching $M$ on $K_n$ with $|M|=k$, one can find $X\in T_k$ such that $\delta(X) \cap M=M$, i.e. such that $M$ induces an injection from $X$ onto $X^c=V_n\setminus X$, that is, $M \subset E(X, V_n\setminus X)$.
\end{cor}
\begin{proof}
Denote $\mathcal{X}_k=\mathcal{P}_k(V_n)$ the set of $k$-subsets of $V_n$. Let $H_k=(\mathcal{X}_k,\mathcal{M}_k)$ be the hypergraph with vertex set $\mathcal{X}_k$, with edge set $\mathcal{M}_k$, and with adjacency matrix $A_{X,M}=1$ if $\delta(X) \cap M=M$, and $A_{X,M}=0$ otherwise. Hence each \emph{(hyper)edge} $M$ \emph{contains} $2^k$ vertices $X\in \mathcal{X}_k$, and each vertex $X$ has degree $\Delta=\frac{(n-k)!}{(n-2k)!}=(n-k) \cdots (n-2k+1) \leq (n-k)^k$. Therefore applying Proposition \ref{lovaszineq} yields the existence of a vertex cover $T_k \subset \mathcal{X}_k$, of size 
$$|T_k| \leq \left(1+\frac{1}{2}+\cdots+\frac{1}{\Delta}\right) \tau^*\leq (1+\ln(\Delta)) \tau^* \leq (1+k\ln(n)) \tau^*$$
with $\tau^*=\inf\{ \sum_{X\in \mathcal{X}_k} t_X | t\in \text{Cov}_k\}$ and where $\text{Cov}_k=\{(t_X)\in \R_{\geq 0}^{\mathcal{X}_k} | \forall M\in \mathcal{M}_k, \sum_X t_X A_{X,M} \geq 1\}$. Letting $t_0\in \R_{\geq 0}^{\mathcal{X}_k}$ be the constant vector with all entries equal to $2^{-k}$, one sees that $t_0\in \text{Cov}_k$, and thus $\tau^*\leq 2^{-k} {n\choose k}$ (the latter is in fact an equality\footnote{Note that if $\pi\in \mathcal{S}_n$ is a permutation of $[n]$, then it induces a bijection on $\mathcal{X}_k$, as well as on $\mathcal{M}_k$. If $t\in \mathcal{R}_{\geq 0}^{\mathcal{X}_k}$, then denote $\pi\cdot t$ the vector with entries $(\pi \cdot t)_X=t_{\pi^{-1}(X)}$. Then $t\in \text{Cov}_k$ if and only if $\pi\cdot t\in \text{Cov}_k$, and $\text{cost}(t)=\text{cost}(\pi\cdot t)=\sum_X t_X$, for any $\pi \in \mathcal{S}_n$. Hence if $t\in \text{Cov}_k$ is such that $\text{cost}(t)=\tau^*$, consider $t^*=(n!)^{-1} \sum_{\pi} \pi\cdot t$ ; one has $t^*\in \text{Cov}_k$ by convexity, and $\text{cost}(t^*)=\text{cost}(t)=\tau^*$. The action $(\pi, t) \mapsto \pi\cdot t$ is transitive on $\text{Cov}_k$ : this shows that $\tau^*$ is reached by some constant vector $t\in \R_{\geq 0}^{\mathcal{X}_k}$, and thus that $\tau^*=2^{-k}{n\choose k}$, for this $H_k=(\mathcal{X}_k,\mathcal{M}_k)$.}). Therefore $|T_k|< (1+k\ln(n)) 2^{-k} {n\choose k}$ as claimed.
\end{proof}
Note that the statement of Corollary \ref{lovaszcoro}, also holds if instead of $G=K_n=(V_n, E_n)$, one considers $G$ an arbitrary graph on $n$ vertices (simply take the same $T_k$ as the one yielded by the above proof). In practice, to find this $T_k$ takes the running time of the greedy algorithm for Vertex Cover, with $H=(V,E)$ a $\Delta$-regular graph ($\Delta=\frac{(n-k)!}{(n-2k)!}$) on $N={n\choose k}$ vertices, thus finding $T_k$ takes $O(N\cdot \Delta \cdot \log(N))=O(n^{2k}k^{-k})$. So finding all $(T_k, 1\leq k\leq \lfloor n/2\rfloor)$ requires $\Theta(n^n)$ time. This means that if one would like to use the protocol of \ref{section:PMG} so as to optimize a convex function over $P_{\text{match}}(n)$ by first optimizing it over a (lifted) convex function over $Q'_{\text{match}}(n)$ (the extension of $P_{\text{match}}(n)$ resulting from the protocol), one shall already have access to a given $(T_k)_{1\leq k \leq n}$, found beforehand, and stored. Otherwise, (i.e. if one must compute the inequalities describing $Q'_n$, and hence first compute some $T_k$) in practice, the gain (in the running time of say the interior point method) supposedly yielded by the use of the compact formulation $Q'_n$, will remain theoretical. The same remark is valid for the lift $Q_n$ of size $\text{poly}(n) 2^{n/2}$ found by \cite{faenzaetal} for the perfect matching polytope.

\subsubsection{Deducing a non-negative factorization from the $4$-round protocol of section \ref{section:PMG}, page 6}
\label{para:protomatchgivesfacto}

Set $\Gamma(\pi)=E\cup V\cup \{\gamma=(\emptyset, (k,j),u,u') : 1\leq k \leq n/2 ; 1\leq j\leq n_k ; u\in V ; u'\in V, u'\neq u\}$.
If $e\in E\subset \Gamma$, set $A_{e',e}=\delta_{e',e}$ for rows $e'$ of $A$, and $A_{*,e}=0$ for other rows. Similarly, if $v\in V\subset \Gamma$, set $A_{v',v}=\delta_{v,v'}$ for rows $v'$ of $A$, and $A_{*,v}=0$ for other rows. And if $\gamma=(\emptyset, (k,j), u,u')\in \Gamma$, then set $A_{x,\gamma}=0$ for all $x\in E\cup V$, and for other rows of $A$:
$$A_{U, \gamma}=p^1_{U,k,j}(u)\omega_{A,U,k,j}(u') \hspace{2mm} \text{with }
\begin{cases}
    p^1_{U,k,j} = \text{Unif}(Z \cap U) \\[1em]
    \omega_{A,U,k,j}(u') = \frac{|U|-1}{2} - \chi_U(u')|Z \cap U|
\end{cases}$$
where, given $(U,k,j)$ one defines $X=X_j\in T_k$,then  $Z=X_j\in T_k$, if $0<|U\cap X_j|\leq \frac{|U|-1}{2}$, and $Z=V_n\setminus X_j$ otherwise. (so $\omega_{A,U,k,j}(.)=\omega_{A,U,|Z\cap U|}(.)$ with $Z=Z(U,k,j)$)

Regarding the matrix $B$, one defines $B_{e,M}=\chi_M(e)$ for $e\in E\subset \Gamma$, $B_{v,M}=1-|M\cap \delta(v)|$ for $e\in V\subset \Gamma$, and, for rows $\gamma=((k,j),u,u')$:
$$B_{\gamma, M}= \mathbf{1}_{|M|=k} p^0_M(k,j) \mathbf{1}_{u'=M(u)} \hspace{2mm} \text{ with }p^0_M(k,.)=Unif(\{j\leq n_k : \delta(X_j)\cap M=M\})$$
(or : $p^0_{M}(k,\cdot)=\delta_{j_0}$ with $j_0$ the least $j\leq n_k$ s.t. $X_j\in T_k$ is  compatible with $M$).

By correctness of the above protocol : $\forall x\in E\cup V\cup \{U\subset V : |U|\geq 3, \text{ odd}\}$:
\[ S_{x,M}=\mathbf{1}_{x\in E} \chi_M(e) +\mathbf{1}_{x\in V} (1-|M\cap \delta(v)|) +\mathbf{1}_{x=U\notin V\cup E} S_{U,M}=\sum_{\gamma\in \Gamma(\pi)} A_{x,\gamma} B_{\gamma,M}\]
Therefore (we recall that $n_k \leq (1+k\ln(n))2^{-k}\binom{n}{k}$) :
\[\text{rk}_+(S)\leq |\Gamma(\pi)|\leq |E|+|V|+ n(n-1)\sum_{k=1}^{\lfloor n/2 \rfloor} n_k \leq n+\frac{n(n-1)}{2} (1+2\sum_k n_k) \leq n^3\ln(n) 1.5^n\]
and hence $\text{xc}(P_{\text{match}}(n))\leq n^3\ln(n) 1.5^n$, by Yannakakis theorem.


\subsection{Equivalent definitions of sorting networks}
\label{section:appendixsorting}
If $x\in \R^n$, we denote $\overline{x}$ the unique vector $y\in \{z\in \R^n : z_1 \leq \cdots \leq z_n\}=\cap_{i<j} H_{i,j}^+$, such that $x$ and $y$ have the same multiset of coordinates.

\begin{prop}
\label{equivdfnsSN}
Let $\mathcal{C}:=((i_0, j_0), \dots, (i_{q-1}, j_{q-1}))$ be a sequence of coordinate pairs in $[n]^2$ with $i < j$. Let $\sigma^+$ be the corresponding composition of conditional reflections:
$ \sigma^+ := \sigma^+_{i_{q-1}, j_{q-1}} \circ \dots \circ \sigma^+_{i_0, j_0} $.
Then the following conditions are equivalent:
\begin{enumerate}
    \item[(i)] $\forall x \in \R^n, \sigma^+(x) = \overline{x}$
    \item[(ii)] $\forall \sigma \in \mathcal{S}_n, \sigma^+(x_\sigma) =x_{\text{id}}=(1, 2, \dots, n)$, where we recall $x_{\sigma} =(\sigma(1), \cdots,\sigma(n))$ ;
    \item[(iii)] $\forall J \subset [n], \sigma^+(\chi_J) = \chi_{J_f}$, where $J_f = \{n-|J|+1, \dots, n\}$.
\end{enumerate}
If $\sigma^+$ satisfies $(i) / (ii)/(iii)$, then $\sigma^+$ is called a \emph{sorting network} (for $\R^n$).
\end{prop}

\begin{proof}
\textbf{(i) $\Rightarrow$ (ii):} This is immediate since $\overline{x_{\sigma}}=x_{\text{id}}$ for any $\sigma \in \mathcal{S}_n$.

\medskip
\noindent \textbf{(ii) $\implies$ (iii):} Let $J \subset [n]$ and let $k = |J|$, we may assume $1\leq k\leq n-1$. We define a permutation $\sigma \in \mathcal{S}_n$ such that $\sigma$ maps the indices in $J^c$ to $\{1, \dots, n-k\}$ and the indices in $J$ to $\{n-k+1, \dots, n\}$. 

Let $f: \R \to \{0,1\}$ be the threshold function defined by $f(u) = 1$ if $u \geq n-k+\frac{1}{2}$ and $f(u) = 0$ otherwise. By construction, we have $f(x_\sigma) = \chi_J$. Since $f$ is non-decreasing, it commutes with each conditional reflection, and thus with $\sigma^+$. Assuming (ii), we have:
\[ \sigma^+(\chi_J) = \sigma^+(f(x_\sigma)) = f(\sigma^+(x_\sigma)) = f(x_{\text{id}}) \]
Since $x_{\text{id}} = (1, \dots, n)$, the function $f$ maps the first $n-k$ coordinates to $0$ and the last $k$ coordinates to $1$. Thus, $f(x_{\text{id}}) = \chi_{\{n-k+1, \dots, n\}} = \chi_{J_f}$, which proves (iii).

\medskip
\noindent \textbf{(iii) $\implies$ (i):} Suppose there exists $x \in \R^n$ such that $z = \sigma^+(x)$ is not sorted. There must exist an index $k$ such that $z_k > z_{k+1}$. Let $u_0 = (z_k + z_{k+1})/2$ and define the threshold function $f(a) = 1$ if $a \geq u_0$ and $f(a) = 0$ otherwise. 
Setting $\chi_J = f(x)$, the commutation property yields:
\[ \sigma^+(\chi_J) = \sigma^+(f(x)) = f(\sigma^+(x)) = f(z) \]
The vector $f(z)$ has a $1$ at index $k$ and a $0$ at index $k+1$. Such a vector is not sorted, i.e., $f(z) \neq \chi_{J_f}$, which contradicts (iii).
\end{proof}

 Denote $\tau_0$ the isometry of $\R^n$ given by $\tau_0 \cdot x=(x_n, \ldots , x_1)$, for $x=(x_1, \ldots ,x_n)\in \R^n$. By extension, we denote $\tau_0 \cdot \sigma=\sigma \circ \tau_0$, the permutation defined by $x_{\tau_0 \cdot \sigma}=\tau_0 \cdot x_{\sigma}$, and we denote $\tau_0(J)$ the $|J|$-subset of $[n]$ whose characteristic vector is given by $\chi_{\tau_0(J)}:=\tau_0 \cdot \chi_J$. The same proof as above (but with threshold functions of type $f(x)=\mathbb{1}_{x\leq a}$) gives a similar statement for \emph{reverse sorting networks}.

\begin{prop}
\label{equivdfnsRSN}
Let $(i_0, j_0), \dots, (i_{q-1}, j_{q-1})$ be a sequence of coordinate pairs in $[n]^2$ with $i < j$. Let $\sigma^-$ be the composition of conditional reflections (onto the $H_{i,j}^-$): $\sigma^- := \sigma^-_{i_{q-1}, j_{q-1}} \circ \dots \circ \sigma_{i_0, j_0}^-$.
Then conditions (i),(ii), (iii) are equivalent: we call $\sigma^-$ a \emph{reverse} sorting network when they are satisfied.
\begin{enumerate}
    \item[(i)] $\forall x \in \R^n, \sigma^-(x) = \tau_0 \cdot\overline{x}$
    \item[(ii)] $\forall \sigma \in \mathcal{S}_n, \sigma^-(x_\sigma) =x_{\tau_0 \cdot\text{id}}=(n, \cdots, 2,1)$ ;
    \item[(iii)] $\forall J \subset [n], \sigma^-(\chi_J) = \chi_{[|J|]}=\chi_{\{1,2, \ldots ,|J|\}}$.
\end{enumerate}
\end{prop}
NB : note that $\{1,2, \ldots ,|J|\}=\tau_0([|n-|J|+1,n|])=\tau_0(J_f)$ (with $J_f$ from Proposition \ref{equivdfnsSN}).
\newcommand{\one}{\mathbf{1}}

\begin{lem}
\label{forwardisreverseandversa}
Let $\sigma^+$ and $\sigma^-$ be the forward and reverse operators associated with a sequence of pairs $(i_k, j_k)_{k=0}^{q-1}$. Then $\sigma^+$ is a sorting network if and only if $\sigma^-$ is a reverse sorting network.
\end{lem}


\begin{proof}
Let $\one = (1, \dots, 1) \in \R^n$, so that for any $J \subset [n]$, one has $\chi_{J^c} = \one - \chi_J$. 

Observe that for ay pair of indices $i < j$ and for any $x \in [0,1]^n$, we have the identity:
\begin{equation}
\label{eq:dualcond}
\sigma_{i,j}^-(\one - x) = \one - \sigma_{i,j}^+(x)
\end{equation}
By composing the $q$ comparators of the network and applying \eqref{eq:dualcond} repeatedly, we obtain:
\[ \sigma^-(\chi_{J^c}) = \sigma^-(\one - \chi_J) = \one - \sigma^+(\chi_J) \]
\begin{equation}
\label{eq:dualcondbis}
\text{In terms of sets, the latter rewrites : }\hspace{2mm} \sigma^-(J^c) = (\sigma^+(J))^c.
\end{equation}

Now the proof easily follows from characterizations $(iii)$ of Propositions \ref{equivdfnsSN} and \ref{equivdfnsRSN}. Indeed, 
\begin{align*}
\sigma^+ \text{ is a SN } \hspace{5mm} & \Leftrightarrow \hspace{2mm} \forall J \subset [n], \sigma^+(\chi_J)=\tau_0 \cdot \chi_{\{1,2\ldots,|J|\}} \\
&  \Leftrightarrow  \hspace{2mm} \forall J \subset [n],\left(\sigma^+(J)\right)^c=\{1,2,\ldots,n-|J|\} \\
&  \Leftrightarrow \hspace{2mm} \forall J \subset [n], \sigma^-(J)=\left(\sigma^+(J^c)\right)^c=\{1,2,\ldots,|J|\}  \hspace{2mm} \text{ by  observation \eqref{eq:dualcondbis} }\\
& \Leftrightarrow \hspace{2mm} \forall J \subset [n], \sigma^-(\chi_J)=\chi_{\{1,2\ldots,|J|\}} \hspace{5mm} \Leftrightarrow\hspace{2mm} \sigma^- \text{ is a reverse SN.}
\end{align*}
\end{proof}

Recall that for a pair $(i,j)\in [n]^2$ such that $i<j$, we have defined (page 11) the operator
$\delta_{i,j}(\sigma,J):=\sigma(J)-\sigma_{i,j}^+(\sigma_{i,j}^-(J))=\langle x_{\sigma}, \chi_J\rangle -\langle \sigma_{i,j}^+(x_{\sigma}), \sigma_{i,j}^-(\chi_J)\rangle$. 
We claim that for any $\sigma\in \mathcal{S}_n$ and any $J\subset [n]$, the following identity holds.

\begin{claim}
\label{neededclaimsorting}
$\forall \sigma, J :\hspace{2mm} \delta_{i,j}(\sigma,J)=\mathbf{1}_{j\in J, i\notin J} (\sigma(j)-\sigma(i))_+ +\mathbf{1}_{i\in J, j\notin J} (\sigma(i)-\sigma(j))_+ $.
\end{claim}
\begin{proof}
The proof is a simple case analysis. It is easier to prove a seemingly more general statement, namely, for $x\in \R^n$ define $\delta_{i,j}(x,J):=\langle x,\chi_J\rangle-\langle \sigma_{i,j}^+(x), \sigma_{i,j}^-(\chi_J)\rangle$ and we claim 
\begin{equation}
\label{eqn:deltaijrewrited}
\forall J\subset [n], \forall x\in \R^n, \quad \delta_{i,j}(x,J)=\mathbf{1}_{j\in J, i\notin J} (x_j-x_i)_+ +\mathbf{1}_{i\in J, j\notin J} (x_i-x_j)_+  
\end{equation}
from which Claim \ref{neededclaimsorting} clearly follows.

If $\{i,j\}\cap J=\emptyset$, then $\sigma_{i,j}^-(J)=J$, and for any $x\in \R^n$ 
$\delta_{i,j}(x,J)=\langle x-\sigma_{i,j}^+(x),\chi_J\rangle =0$ since $(\sigma_{i,j}^+(x))_k=x_k$ for all $k\in [n]\setminus \{i,j\}$.

If $\{i,j\} \subset J$, then one also has $\sigma_{i,j}^-(J)=J$, and for $x\in \R^n$:
$\delta_{i,j}(x,J)=\langle x-\sigma_{i,j}^+(x),\chi_J\rangle=\langle x-\sigma_{i,j}^+(x),e_i+e_j\rangle=0$ (set $y=\sigma_{i,j}^+(x)$ : the second equality is also because $y_k=x_k$ for all $k\neq i,j$, and the last equality is due to the fact that $\{x_i,x_j\}=\{y_i,y_j\}$ (for any $x\in \R^n$).

Now if $j\in J$ but $i\notin J$, one has $\sigma_{i,j}^-(J)=J\triangle \{i,j\}=(J\setminus \{j\})\cup \{i\}$ and so, for $x\in \R^n$:
$\delta_{i,j}(x,J) 
= \langle x,e_j \rangle - \langle \sigma_{i,j}^+(x),e_i \rangle 
= \begin{cases} 
    0 & \text{ if } x \in H_{i,j}^- \\ 
    x_j - x_i & \text{ if } x \in H_{i,j}^+ 
  \end{cases}$ , so $\delta_{i,j}(x,J)=(x_j-x_i)_+$ for all $x\in \R^n$.

  Finally, if $j\notin J$ but $i\in J$, then $\sigma_{i,j}^-(J)=J$, and
  $\delta_{i,j}(x,J)=\langle x-\sigma_{i,j}^+(x), e_i\rangle=
  \begin{cases} 
    0 & \text{ if } x \in H_{i,j}^+ \\ 
    x_i - x_j & \text{ if } x \in H_{i,j}^- 
  \end{cases}$ and so $\delta_{i,j}(x,J)=(x_i-x_j)_+$ for all $x\in \R^n$.

  Putting the four cases together yields \eqref{eqn:deltaijrewrited}, and hence the claim.
\end{proof}

\subsection{A remark about minimal sorting networks}
\label{para:minimalSN}

Let $\mathcal{C}=((i_0,j_0), \ldots ,(i_{q-1},j_{q-1}))$ be a sequence of pairs of coordinates, yielding a sorting network $\sigma^+=\sigma_{\mathcal{C}}^+:=\sigma_{i_{q-1},j_{q-1}}^+ \circ \cdots \circ \sigma_{i_0,j_0}^+$. Let us call $\mathcal{C}$ minimal if for any $\mathcal{D}$ obtained from $\mathcal{C}$ by removing some of the pairs $(i,j)$ (and not modifying the order), $\sigma_{\mathcal{D}}^+$ isn't a sorting network. One may equivalently define minimality of $\mathcal{C}$ via the reverse sorting network $\sigma^-$ defined following the same sequence of comparators.

Saying $\mathcal{C}$ is minimal in the above sense doesn't imply that $q=|\mathcal{C}|$ is the least possible length of a sorting network. For instance consider $\mathcal{C}=\{(1,n),(1,n-1), \ldots, (2,n), (2,n-1), \ldots, (n-1,n)\}$, with $q=\frac{n(n-1)}{2}$. It is easily seen that $\mathcal{C}$ is a minimal SN : if one remove say $(k,k+u)$ (with $1\leq k \leq k+u \leq n$) from $\mathcal{C}$ (and possibly more $(i,j)$ are removed), then $J:=[k-1] \cup \{k+u\}$ will not be sorted by any $\sigma^+_{\mathcal{D}}$ with $\mathcal{D}\subset \mathcal{C} \setminus \{(k,k+u)\}$.

Assuming the sorting network $\sigma^+$ provided to Alice and Bob was minimal, and denoting $A, B$, the non-negative matrices resulting from the protocol, note that none of the $2q$ columns of $A$ and none of the rows of $B$, is all-zero. Indeed, since $B_{l,1,\sigma}=(\sigma_l(j_l)-\sigma_l(i_l))_+$, we have at least $B_{l,1,\text{id}_{[n]}}=j_l-i_l >0$ (since $\sigma_l=id$ if $\sigma_0=id$, no matter what the comparators $c_k=(i_k,j_k)$, $0\leq k \leq l-1$ are), so $B_{l,1,*}$ is never all-0s. Similarly, if $J_0=[i_l]=\{1, \ldots , i_l\}$, then $J_l=\ldots=J_0$, and so $A_{J_0, l,-1}=1_{i_l \in J_l}1_{j_l\notin J_l}=1$, while choosing $J_0=[j_l]$ gives $A_{J_0,l,-1}=0$, so columns $A_{*,l,-1}$ are never all $0s$, nor all-$1s$.
If we had $B_{l,-1,\sigma}=0$ for all $\sigma$, for some $l\in [|0,q-1|]$, it would mean that $\sigma_l(j_l)>\sigma_l(i_l)$ (for any $\sigma$, with $\sigma_l=\sigma^+_l(\sigma):=\sigma_{i_{l-1},j_{l-1}}^+ \circ \cdots \circ \sigma_{i_0,j_0}^+(\sigma)$), i.e. that $\sigma_{i_l,j_l}^+$ is always acting as $id_{\R^n}$, i.e. that $\mathcal{C}\setminus \{c_l\}$ already defines a SN (which would contradict minimality of $\mathcal{C}$).

Similarly, if we had $A_{J,l,1}=0$ for all $J$, this would mean that for all $J$, $i_l\in J_l$ or $J_l\cap \{i_l,j_l\}=\emptyset$, and in particular that while running the reverse SN $\sigma^-$ starting from some $\chi_J$, the map $\sigma_{i_l,j_l}^-$ always acts as $\text{id}_{\R^n}$, and so that the comparator $c_l$ could be removed, meaning one would stil have $\sigma^-(\chi_J)=\chi_{[|J|]}$ for all $J\subset [n]$, with $\sigma^-=\sigma^-_{\mathcal{C}\setminus \{c_l\}}$, hence contradicting minimality (by \ref{equivdfnsSN}).

Clearly, if $\min\{\text{rk}_+(A), \text{rk}_+(B)\}=m<2q$, then the (non-neg.) factorization $S=AB$, together with the further factorization of $A$ (or of $B$) gives a non-negative factorization $S=A'B'$ of size $m$ only. Hence, further than showing that no columns of $A$ is all-0s (nor any row of $B$), one would hope to prove that minimality of the sorting network, ensures  that $\text{rk}_+(A)=2q$ and that $\text{rk}_+(B)=2q$. The latter seems a rather rigid statement, however we conjecture that if $\sigma^+$ is minimal, and $A,B$ are the matrices produces by the above protocol, then $\min\{\text{rk}_+(A), \text{rk}_+(B)\} \geq q$.

For instance, with $\mathcal{C}=[(1,n),(1,n-1), \ldots, (2,n), \ldots ,(2,3), .. ,(n-1,n)]$, a sequence of $q=\frac{n(n-1)}{2}$ pairs, denote $\tau_{<(i,j)}^-$ the operator onto the subsets $J$, in other words : $\tau_{<(i,j)}^-(J)=K$ if $\chi_K=\sigma_{i_{\ell-1},j_{\ell-1}}^- \circ \cdots \circ \sigma_{i_0,j_0}^-(\chi_J)$, with 
$\ell=\ell(i,j)=i(n-i)+\frac{i(i+1)}{2}-j$
the position of comparator $c_{i,j}$ withiin this $\mathcal{C}$. Then let us denote $A_{J,(i,j),1}=1_{j\in J_l} 1_{i\notin J_l}$ and $A_{J,(i,j),-1}=1_{i\in J_l} 1_{j\notin J_l}$, corresponding to the entries $A_{J,l, \varepsilon}$, with $l=l(i,j)$. Thus $A$ is an $\mathcal{J} \times [2q]=\mathcal{J}\times [n(n-1)]$ $0$-$1$ matrix, and if we set (for every $i<j$):
$$J_{i,j,+}=[j]\setminus \{i\} \text{ and }J_{i,j,-}=[|1,i-1|] \cup \{j+1\} \text{ if } \hspace{2mm} j<n, \text{and } J_{i,n,-}=[i] \hspace{2mm} (\text{ for }j=n),$$
then it is not hard to check that $\mathcal{F}=\mathcal{F}^+\cup \mathcal{F}^-$ is a fooling set of size $2q$, with $\mathcal{F}^+=\{(J_{i,j,+},(i,j,1))\}$ and $\mathcal{F}^-=\{(J_{i,j,-},(i,j,-1))\}$. Indeed $\tau^-_{<(i,j)}(J_{i,j,+})=J_{i,j,+}=[j]\setminus \{i\}$ so $A_{J_{x,+},(x,+1)}=1$ for all $x\in \mathcal{P}=\{(i,j)\in [n]^2 : i<j\}$, while $\tau^-_{<(i,j)}(J_{i,j,-})=[i]$ and so $A_{J_{x,-},(x,-1)}=1$ for all $x$. Morever if $x=(i,j)<y=(u,v)$, i.e. if $u>i$ or $u=i$ and $v>j$, then $A_{J_{x,+}, (y,+)}=0$ [ because either $u\leq j-1$ (and so $u\in \tau_{<(u,v)}^-(J_{x,+})=[j-1]$), or $u \geq j$ thus $v>u \geq j$ ensuring $v\notin [j-1]$ (and thus $A_{J_{x,+}, (y,+)}=0$)].%
Similarly, if $x<y$, then $A_{J_{x,-}, (y,-)}=0$. Say $x=(i,j)$ and $y=(u,v)$. Then, if $u>i$, one has $\tau^-_{<(u,v)}(J_{x,-})=[i]$, so $u\notin [i]$ implies $A_{J_{x,-}, (y,-)}=0$, while if $u=i$ but $v>j$, then $i\notin J_{x,-}=\tau^-_{<(i,v)}(J_{x,-})$, so $A_{J_{x,-}, (y,-)}=0$.
Moreover, one easily checks that  $A_{J_{x,-},(y,+)}=0$ for any pairs $x,y$, except for $x=(i,j)$ et $y=(i,j+1)$ (with $j<n$). But for this pair, one has $A_{J_{y,+},(x,-)}=A_{[j],(i,j,-)}=0$.
This shows that $\mathcal{F}$ is a fooling set of $A$, so that $\text{rk}_+(A)\geq \text{cov}(A) \geq 2q$, showing that $\text{rk}_+(A)=2q=n(n-1)$, in this particular example.

Regarding the matrix $B$, we may divide it into two horizontal submatrices, with the first $q$ rows forming the matrix $B_+=(B_{(i,j,+), \sigma})_{(i,j)\in \mathcal{P} ; \sigma\in \mathcal{S}_n}$, and the last $q$ rows form the submatrix $B_-:=(B_{(i,j,-), \sigma})_{(i,j)\in \mathcal{P} ; \sigma\in \mathcal{S}_n}$ : note that unlike for $A$, one here has $B_-=\textbf{1} -B_+$, with $\textbf{1}$ the $q\times n!$ matrix with all entries equal to $1$. The $(q\times q)$-submatrix of $B_-$ obtained by keeping the rows corresponding to transpositions $\tau_{i,j}$, is the diagonal matrix, showing that $rk_+(B_-)=rk_(B_-)=q$, and that $rk_+(B_+)\geq rk(B_+) \geq q-1$. However, proving a (conjectural) lower bound $\text{rk}_+(B) \geq 2q-1$ seems already quite harder than for $A$.

In fact, this becomes of interest when asking for the exact value of $\text{xc}(\text{Perm}(n))$: thanks to Goemans (\cite{goemans}), we know that the extension complexity of the permutahedron is of order $n\ln(n)$, and more precisely that $\log_2(n!) \leq \text{xc}(\text{Perm}(n)) \leq 2q \leq C n \log(n)$ (by taking a minimal sorting network, in $q$ comparators, given by AKS), and conjecturally $\text{xc}(\text{Perm}(n)) \in [|q,2q|]$, if $q$ is minimal. The hard part of this conjecture, is to show that  $\text{xc}(\text{Perm}(n))\geq q$. Heuristically, it says that sorting networks give the best lifts of the permutahedron. The question of where between $q$ and $2q$, lies the actual size of the extension yielded by above factorization $S=AB$, comes naturally, by analogy with the case of the regular $n$-gon, for which $\text{xc}(P_n)\in \{2q, 2q-1\}$ (see \cite{szusterman}).

\subsection{About Goemans extension of the permutahedron}
\label{para:comparisongoemans}

Let $\sigma^+=\sigma_{i_{q-1},j_{q-1}}^+\cdots \circ \sigma_{i_0,j_0}^+$ be a sorting network in $q$ comparators.

Given $x=x_{\sigma}=(\sigma(1), ... , \sigma(n))$, we can let $x$ go through $\sigma^+$, and this yields a sequence of $q+1$ vectors in $\R^n$ : $x^{(0)}=x$, $x^{(1)}=\sigma_{i_0,j_0}^+(x)$, etc. , $x^{(k+1)}=\sigma_{i_k,j_k}^+(x^{(k)})$, up to $x^{(q)}=\sigma_{i_{q-1},j_{q-1}}^+(x^{(q-1)})$. In particular, note that, (since $\sigma^+$ is a sorting network) :
\[ x^{(q)}=\sigma_{i_{q-1},j_{q-1}}^+ \circ \cdots \circ \sigma_{i_0,j_0}^+(x_{\sigma})=\sigma^+(x_{\sigma})=x_{id}=(1,2,... ,n).\]
Moreover, denote $\theta_k=\frac{1}{\sqrt{2}}(e_{j_k}-e_{i_k}) \in \mathbb{S}^{n-1}$, and $\pi_k=\pi_{i_k,j_k}$ the orthogonal projection onto $H_{i_k,j_k}:=\theta_k^{\perp}$. Since $x^{(k+1)}=\sigma_{i_k,j_k}^+(x^{(k)})$, we have :
\begin{align*}
\pi_k(x^{(k+1)})=\pi_k(x^{(k)}) ; \quad 
\langle x^{(k+1)}, \theta_k \rangle \geq |\langle x^{(k)}, \theta_k \rangle| \hspace{2mm}\text{ i. e.}\langle x^{(k+1)} \pm x^{(k)}, \theta_k \rangle \geq 0.
\end{align*}
(the two inequalities hold with equality here).

Denote $w_{\sigma}:=(x^{(0)},x^{(1)}, \cdots, x^{(q)}) \in \R^{n(q+1)}$, with $x^{(0)}=x_{\sigma}$ (and the $x^{(k)}$ as above). Then clearly $w_{\sigma} \in Q_n$, where $Q_n$ is the following polytope:
\[Q_n=\{w=(y_0, ... , y_q)\in \R^{n(q+1)} | y_q=x_{id} ; \pi_k(y_{k+1}-y_k)=0 ; \langle y_{k+1} \pm y_k, \theta_k \rangle \geq 0\} \]

Hence if $\pi : \R^{n(q+1)} \to \R^n , (y_0, ... ,y_q)\mapsto y_0$, is the projection onto first $n$ coordinates, we see that $\pi w_{\sigma}=x_{\sigma}$, and hence that $\text{Perm}(n) \subset \pi Q_n$.
In fact, this is an equality, i.e. $Q_n$ is an extension of $\text{Perm}(n)$, as shown by Goemans.
\begin{prop}[Goemans]
$\pi Q_n=\text{Perm}(n)$
\end{prop}
\begin{proof}
One can use Edmonds representation of $\text{Perm}(n)$: so as to check that $y=\pi(w) \in \text{Perm}(n)$, for any $w\in Q_n$, one needs to check that $y([n])=\frac{n(n+1)}{2}$, and that $y(J)\geq \frac{|J|(|J|+1)}{2}$, for $y=y_0$.

Since $e_i+e_j \in H_{i,j}$ and $e_l \in H_{i,j}$ (for $l\neq i,j$), observe that $\pi_k(y_{k+1}-y_k)=0$ implies that $y_k([n])=y_{k+1}([n])$. In other words, if $w\in Q$, then $y_0([n])=...=y_q([n])$ and since $y_q=x_{id}$, one gets that $y_0([n])=\frac{n(n+1)}{2}$.

Similarly, $y_q(J)=x_{id}(J)=\frac{|J|(|J|+1)}{2}$, so it suffices to argue that (when $w=(y_l)_{l=0}^q \in Q_n$), for any $1\leq k \leq n$, the sequence $m_k(y_j):=\min_{|J|=k} y_j(J)$, is non-increasing in $j\leq q$. 

Let $y\in \R^n$ and let $z\in \R^n$ be such that $\pi_{i,j}(y)=\pi_{i,j}(z)$, and $\langle z\pm y, \theta_{i,j} \rangle \geq 0$. This exactly means that $y_l=z_l$ for $l\neq i,j$, that $y_i+y_j=z_i+z_j$, and that $z_j \geq \max \{y_i, y_j\}$. Since $m_k(y)$ depends only on $\overline{y}$, to check that $m_k(y)\geq m_k(z)$, we may assume that $y=\overline{y}=(u_1, ... ,u_n)$ with $u_a\leq u_{a+1}$ for $a=1,2, ... , n-1$.
Thus $z=(u_1, ... ,u_{i-1},u^-,u_{i+1}, ... ,u_{j-1}, u^+, u_{j+1}, ... ,u_n)$, with $u^-+u^+=u_i+u_j$, and $u^- \leq u_i$ (and $u^+\geq u_j$). If $k\leq j-1$, then it is clear\footnote{since the $a^{th}$-smallest coordinate of $z$ will be at most the $a^{th}$ coordinate of $y$, for all $a\leq j-1$} that $m_k(z)\leq m_k(y)$. Say $k=j$ and $u^+>u_j$, then
\[m_k(z)=m_k(y)+u^--u_i+\min\{ u^+, u_{j+1}\} -u_j \leq =m_k(y)+(u^++u^-)-(u_i+u_j)=m_k(y) \]
and similarly, if $j<k\leq n$, one see that, either $u^+>u_{k+1}$ and so
\[m_k(z)=m_k(y)+u^--u_i+u_{k+1}-u_j < m_k(y)+u^--u_i+u^+-u_j=m_k(y)\]
or that $u^+ \leq u_{k+1}$, so that $m_k(z)=m_k(y)$.

This shows that $(m_k(y_j) , 0\leq j\leq q)$ is non-increasing, and thus that $\min_{|J|=k} y_0(J)=m_k(y_0) \geq m_k(y_q)=\frac{k(k+1)}{2}$, i.e. as needed.
\end{proof}

One may argue that there is a lot of redundancy in the above $Q_n$. Indeed, if $w=(y_k)_{k=0}^q \in Q_n$, then $(y_k)_l=(y_{k+1})_l$ for all $l\notin \{i_k, j_k\}$. Hence instead of giving all of $y_{k+1}$, which seems repetitive, one may prefer to only give $(a_{k+1},b_{k+1}):=((y_{k+1})_{i_k}, (y_{k+1})_{j_k})$. This defines a projection of the above $Q_n$: $\tilde{\pi}(Q_n)=\tilde{Q}_n$, where we keep $y_0$, and then we only record the $i_{k-1}$ and $j_{k-1}$ coordinates of $y_k$ (i.e. $\tilde{\pi}(y_0, ... , y_q)=(y_0,a_1,b_1, ... ,a_q,b_q)\in \R^{n+2q}$).
In other words, set
\[
\tilde{Q}_n := \left\{
    \tilde{w} \in \mathbb{R}^{n+2q} :
    z_q(\tilde{w})=x_{id}, \quad
    \begin{aligned}
        & a_j+b_j=\alpha_j(\tilde{w})+\beta_j(\tilde{w}) \\
        & b_j \geq \max\{ \alpha_j(\tilde{w}),\beta_j(\tilde{w})\}
    \end{aligned}
    \quad (\forall 1\leq j\leq q)
\right\}
\]
where the linear forms $\alpha_j,\beta_j$ and the linear function $z_q$ are defined as follows: first define recursively $z_1(\tilde{w}), ... , z_q(\tilde{w}) \in \R^n$, i.e. $z_1(\tilde{w})$ is the $z_1\in \R^n$ s.t. $(z_1)_l=(y_0)_l$ for $l\neq i_0, j_0$, and such that $(z_1)_{i_0}=a_1$ and $(z_1)_{j_0}=b_1$ ; then $z_2(\tilde{w})$ is the $z_2\in \R^n$ s.t. $(z_2)_l=(z_1)_l$ for $l\neq i_1, j_1$, and such that $(z_2)_{i_1}=a_2$ and $(z_2)_{j_1}=b_2$, etc. (it is easily \footnote{in fact each $z_m$, $1\leq m\leq q$, is the projection which keeps exactly $n$ of the $n+2q$ coordinnates (more precisely, of the first $n+2m$ coordinates) of $\tilde{w}$, and permute them in a certain order. Which coordinates, and what order, is determined by the first $m$ pairs of $\sigma^+$ : $(i_0,j_0), ... , (i_{m-1},j_{m-1})$.} checked that each of the maps $z_j : \R^{n+2q}\to \R^n$, is linear). Also define $z_0(\tilde{w})=y_0$. Then $\alpha_m(\tilde{w})$ is the $i_{m-1}^{th}$ coordinate of $z_{m-1}(\tilde{w})$, and $\beta_m(\tilde{w})$ is the $j_{m-1}^{th}$ coordinate of $z_{m-1}(\tilde{w})$. 

\begin{claim}
$Q_n$ and $\tilde{Q}_n$ are affinely equivalent : there exist affine maps $\tilde{\pi} : \R^{n(q+1)} \to \R^{n+2q}$, and $\overline{\pi}:\R^{n+2q}\to \R^{n(q+1)}$, such that 
$\overline{\pi} \tilde{Q}_n=Q_n $, and
$\tilde{\pi} Q_n=\tilde{Q}_n$.
\end{claim}

\begin{proof}[Proof of Claim] we first argue that $\tilde{Q}_n=\tilde{\pi} Q_n$, with $\tilde{\pi}$ the projection as before. In fact, if $w=(y_0, ... ,y_q)\in Q_n$, and if $\tilde{w}:=\tilde{\pi} w$, then the fact that $\pi_k(y_{k+1})=\pi_k(y_k)$ gives us that $a_{k+1}+b_{k+1}=\alpha_{k+1}(\tilde{w})+\beta_{k+1}(\tilde{w})$, and the fact that $\langle y_{k+1} \pm y_k,\theta_k \rangle \geq 0$, tells us exactly that $b_{k+1}\geq \max \{\alpha_{k+1}(\tilde{w}), \beta_{k+1}(\tilde{w})\}$. Moreover by construction, $z_q(\tilde{w})=z_q(\tilde{\pi} w)=y_q$, and so $z_q(\tilde{w})=x_{id}$.
So $\tilde{\pi} Q_n \subset \tilde{Q}_n$. Conversely, define 
$\overline{\pi}(\tilde{w}):=(y_0, z_1(\tilde{w}), \cdots ,z_q(\tilde{w})).$ Then $\overline{\pi}(\tilde{w})\in \R^{n(q+1)}$, and similarly as before, the equations and inequalities satisfied by $\tilde{w}$, translate exactly into the affine constraints defining $Q_n$, so that $\overline{\pi}(\tilde{Q}_n) \subset Q_n$.
Therefore $\tilde{Q}_n \subset \tilde{\pi} Q_n$, because $\tilde{\pi} \circ \overline{\pi} (\tilde{w})=\tilde{w}$.
And similarly, $\overline{\pi}(Q_n)=Q_n$, because if $w\in Q_n$, then $w=\overline{\pi} \circ \tilde{\pi}(w)$.
\end{proof}

\begin{rem}
$\tilde{Q}_n$ is the extension defined in Goemans' article \cite{goemans}, but since $Q_n$ is simply another representation of the same polyhedron, we may also call it \emph{Goemans extension}.
\end{rem}

A natural question, given the protocol on page 12, \ref{para:newprotpermu} (defined after a given sorting network $\sigma^+$ of $\R^n$), is whether $Q_n$ and the extension resulting from the protocol are related, and how.

Say $P=\{x\in \R^n : Dx\geq b, Cx=\lambda \}$ is a polytope, and $S$ its slack matrix. Assume $S=AB$ is a non-negative factorization of $S$, of size $r \geq 1$. We recall (see for instance \cite{yannak91}) that the extension of $P$ yielded by this factorization is :
$$Q:=\{(x,y)\in \R^{n+r} : Dx-Ay=b, Cx=\lambda, y\geq 0\}.$$

Denote $\mathcal{J}=\mathcal{P}([n])\setminus \{\emptyset, [n]\}$, $b_J=\frac{|J|(|J|+1)}{2}$, $D_J=\chi_J^{\top}$, and $\lambda=\frac{n(n+1)}{2}$ and $C=\chi_{[n]}^{\top}$ (so that $P=\{x : Dx\geq b, Cx=\lambda\}=\text{Perm}(n)$).
A closer look at the extension $Q_{\sigma^+}=\{(x,z)\in \R^n \times \R_{\geq 0}^{2q} : Dx-Ay=b , Cx=\lambda\}$, with $(D,b,C,\lambda)$ as before, and with $A$ the matrix from the factorization $S=AB$ given on page 10, i.e. $A_{J, l,+}=\mathbf{1}_{j_l\in J_l, i_l\notin J_l}$ and $A_{J, l,-}=\mathbf{1}_{i_l\in J_l, j_l\notin J_l}$, with $l\in [|0,q-1|]$, reveals that $Q_n$ (or $\tilde{Q}_n$) can be seen, up to affine equivalence, as a section of $Q_{\sigma}^+$. This was proven in (Proposition 15, p. 43, \cite{szusterman}).

In this sense, Goemans extension is a better one : $\text{dim}(Q_n)\leq \text{dim}(Q_{\sigma^+})$ and $|Q_n| \leq |Q_{\sigma^+}|$. It is open whether any of these two inequalities is strict or not.
A more interesting question is to investigate whether $|Q_n| \geq q$, or at least, whether $|Q_{\sigma}^+| \geq q$, when $\sigma^+$ is a minimal sorting network in $q$ comparators (see paragraph \ref{para:minimalSN} for a definition of minimality, and for an example of quadratic SN, for which this inequality holds).

\end{document}